\documentclass[12]{article}
\usepackage[latin1]{inputenc}
\usepackage{hyperref}
\usepackage{amsmath}
\usepackage{latexsym}
\usepackage{times}
\usepackage{cite}
\usepackage{amssymb}
\usepackage{amsfonts}
\usepackage{amscd}
\usepackage{xcolor}
\newtheorem{theorem}{Theorem}

\newtheorem{proof}[theorem]{Proof}

\newtheorem{example}[theorem]{Example}

\newtheorem{lemma}[theorem]{Lemma}

\newtheorem{proposition}[theorem]{Proposition}

\newcommand{\cqfd}{\hfill $\square$}

\begin{document}
\begin{center}
\Large{$\mathcal{R}(p,q)-$ deformed super Virasoro $n-$ algebra}\\
\vspace{0,5cm}
Fridolin Melong\\
\vspace{0.25cm}



{\em Institut f\"ur Mathematik, Universit\"at Z\"urich,\\Winterthurerstrasse 190, CH-8057, Z\"urich, Switzerland\\
	fridomelong@gmail.com\\
}
\end{center}
\begin{abstract}
In this paper, we construct the super Witt algebra and super Virasoro algebra in the framework of the $\mathcal{R}(p,q)-$ deformed quantum algebras. Moreover, we perform  the super $\mathcal{R}(p,q)-$ deformed Witt $n-$ algebra, the $\mathcal{R}(p,q)-$ deformed Virasoro $n-$ algebra and discuss the super $\mathcal{R}(p,q)-$ Virasoro $n-$ algebra ($n$ even ). Besides, we define and construct another super $\mathcal{R}(p,q)-$ deformed Witt $n-$ algebra and  study a toy model for the super $\mathcal{R}(p,q)-$ Virasoro constraints. Relevant particular cases  induced from the  quantum algebras known in the literature are deduced from the formalism developped.
\end{abstract}

{\bf keyword}
	$\mathcal{R}(p,q)-$ calculus, Super Virasoro algebra, super-Virasoro constraints.

\tableofcontents
\section{Introduction}\label{sec1}

The  nature of the Virasoro algebra  was described by  
Kupershmidt \cite{KB}. Its applications in mathematics and physics, such that in  conformal field theory and string theory  were also presented \cite{BPZ,RY, KB}. Many   generalizations  and deformations (one or two parameters) of the  Virasoro algebra were investigated in the literature\cite{AS,CJ,Hounkonnou:2015laa}.
The generalization of   Kupershmidt's work was provided in  \cite{KB}.
The relation between the Korteweg-de Vries (KdV) equation and the Virasoro algebra was described  by Gervais \cite{G} and Kupershmidt \cite{KB}. Moreover, 
Huang and Zhdanov  presented the realizations of  Witt and Virasoro algebras. Their connection with integrable equations was  determined\cite{HZ}.

The construction of $\alpha^k$ derivation and a representation theory were investigated \cite{AMS}. Also, the cohomology complex of Hom-Lie superalgebras was furnished and the central extensions was computed. As application, the derivations and the second cohomology group of a twisted $osp(1,2)$ superalgebra were calculated. Moreover,  Curtright and Zachos introduced the $q-$ deformed Witt algebra \cite{CZ} and from this results  Ding et {\it al } determined a nontrivial $q-$ deformed Witt $n-$ algebras. It is a generalization of the Lie algebra also called sh-n-Lie algebra \cite{D}.
Wang  et {\it al} \cite{WYLWZ}
investigated the two different $q-$ deformed Witt algebra and constructed their $n-$ algebras. In
one case, the super version is also presented. Moreover the central
extensions is provided and the super $q-$ deformed Virasoro $n-$ algebra for the $n$ even case is furnished.

The two parameters deformation of the Virasoro algebra with conformal dimension was studied in \cite{CJ}. Also, the central charge term for the Virasoro algebra and the associated deformed nonlinear equation (Korteweg-de Vries equation) were determined.

Moreover, the  generalizations of $(p,q)$- deformed Heisenberg algebras,  called $\mathcal{R}(p,q)$- deformed quantum algebras were  investigated in\cite{HB1}.
Hounkonnou and Melong\cite{HM}  constructed the $\mathcal{R}(p,q)$- deformed conformal Virasoro algebra,  derived the $\mathcal{R}(p,q)$- deformed  Korteweg- de Vries equation for a conformal dimension $\Delta=1$, and presented the energy-momentum tensor from the $\mathcal{R}(p,q)$- deformed quantum algebras   for the conformal dimension $\Delta=2$.

Recently, the  generalizations   of  Witt and Virasoro algebras were performed, and  the associated Korteweg-de Vries equations from the $\mathcal{R}(p,q)-$ deformed quantum  algebras  were derived. Related relevant properties were investigated and discussed. Furthermore,  the $\mathcal{R}(p,q)-$ deformed Witt $n-$ algebra constructed, and  the Virasoro constraints for a toy model, which play an important role in the study of matrix models was presented \cite{HMM}.

The aim of this paper is to construct the super Witt $n-$ algebra, Virasoro $2n-$ algebra, and  super Virasoro $n-$ algebra ($n$ even ) from the quantum deformed algebra \cite{HB1}. As application, we construct another super $\mathcal{R}(p,q)-$ deformed Witt $n-$ algebra and investigate a toy model for the super $\mathcal{R}(p,q)-$ Virasoro constraints. Furthermore, we deduce particular cases associated to quantum algebra presented in the literature.

This paper is organized as follows: Section $2$ is reserved to some notations, definitions and results used in the sequel. In section $3,$ we investigate the super Witt algebra and super Witt $n-$ algebra induced by the $\mathcal{R}(p,q)-$ deformed quantum algebra. Moreover, we  construct the  $\mathcal{R}(p,q)-$ deformed Virasoro $2n-$ algebra and deduce particular cases. In section $4,$ we furnished the super ${\mathcal R}(p,q)-$ deformed Jacobi identity. Besides, we construct the super $\mathcal{R}(p,q)-$ deformed Virasoro algebra and perform the super $\mathcal{R}(p,q)-$ deformed Virasoro $n-$ algebra. Particular cases are deduced. Section $5$ is dedicated to the application. We contruct another super $\mathcal{R}(p,q)-$ deformed Witt $n-$ algebra and study a toy model. We end with the concluding remarks in section $6.$

\section{Basics definitions and  notations}\label{sec2}

Let us  recall some definitions, notations, and known results  used in this work.
For that, let $ p$ and $q,$ two positive real numbers such that $ 0<q<p\leq 1,$ and a 
meromorphic function $\mathcal{R}$ defined on $\mathbb{C}\times\mathbb{C}$ by \cite{HB}: \begin{eqnarray}\label{r10}
\mathcal{R}(s,t)= \sum_{u,v=-l}^{\infty}r_{uv}s^u\,t^v,
\end{eqnarray}
where $r_{uv}$ are complex numbers, $l\in\mathbb{N}\cup\left\lbrace 0\right\rbrace,$ $\mathcal{R}(p^n,q^n)>0,  \forall n\in\mathbb{N},$ and $\mathcal{R}(1,1)=0$ by definition. The bidisk $\mathbb{D}_{R}$ is defined by:  \begin{eqnarray*}
	\mathbb{D}_{R}
	&=&\left\lbrace a=(a_1,a_2)\in\mathbb{C}^2: |a_j|<R_{j} \right\rbrace,
\end{eqnarray*}
where $R$ is the convergence radius of the series (\ref{r10}) defined by Hadamard formula \cite{TN}:
$$	\lim\sup_{s+t \longrightarrow \infty} \sqrt[s+t]{|r_{st}|R^s_1\,R^t_2}=1.
$$
We also consider $\mathcal{O}(\mathbb{D}_{R})$ the set of holomorphic functions defined on $\mathbb{D}_{R}.$
Define the  $\mathcal{R}(p,q)-$ deformed numbers  \cite{HB}:
\begin{eqnarray}\label{rpqnumber}
[n]_{\mathcal{R}(p,q)}:=\mathcal{R}(p^n,q^n),\quad n\in\mathbb{N}\cup\{0\},
\end{eqnarray}
the
$\mathcal{R}(p,q)-$ deformed factorials
\begin{eqnarray*}\label{s0}
	[n]!_{\mathcal{R}(p,q)}:=\left \{
	\begin{array}{l}
		1\quad\mbox{for}\quad n=0\\
		\\
		\mathcal{R}(p,q)\cdots\mathcal{R}(p^n,q^n)\quad\mbox{for}\quad n\geq 1,
	\end{array}
	\right .
\end{eqnarray*}
and the  $\mathcal{R}(p,q)-$  binomial coefficients
\begin{eqnarray*}\label{bc}
	\bigg[\begin{array}{c} m  \\ n\end{array} \bigg]_{\mathcal{R}(p,q)} := \frac{[m]!_{\mathcal{R}(p,q)}}{[n]!_{\mathcal{R}(p,q)}[m-n]!_{\mathcal{R}(p,q)}},\quad m,n\in\mathbb{N}\cup\{0\},\quad m\geq n.
\end{eqnarray*}
Consider the following linear operators defined on  $\mathcal{O}(\mathbb{D}_{R}),$ (see \cite{HB1} for more details),
\begin{eqnarray*}
	\;Q:\psi\longmapsto Q\psi(z):&=& \psi(qz),\\
	\; P:\psi\longmapsto P\psi(z):&=&\psi(pz),\\
	\;\partial_{p,q}:\psi\longmapsto \partial_{p,q}\psi(z):&=&\frac{\psi(pz)-\psi(qz)}{z(p-q)},
\end{eqnarray*}
and the $\mathcal{R}(p,q)-$ derivative 
\begin{eqnarray*}\label{r5}
	\partial_{\mathcal{R}( p,q)}:=\partial_{p,q}\frac{p-q}{P-Q}\mathcal{R}( P,Q)=\frac{p-q}{p^{P}-q^{Q}}\mathcal{R}(p^{P},q^{Q})\partial_{p,q}.
\end{eqnarray*}
The  algebra associated with the $\mathcal{R}(p,q)-$ deformation is a quantum algebra, denoted $\mathcal{A}_{\mathcal{R}(p,q)},$ generated by the set of operators $\{1, A, A^{\dagger}, N\}$ satisfying the following commutation relations \cite{HB1}:
\begin{eqnarray*}
	&& \label{algN1}
	\quad A A^\dag= [N+1]_{\mathcal {R}(p,q)},\quad\quad\quad A^\dag  A = [N]_{\mathcal {R}(p,q)}.
	\cr&&\left[N,\; A\right] = - A, \qquad\qquad\quad \left[N,\;A^\dag\right] = A^\dag
\end{eqnarray*}
with the realization on  ${\mathcal O}(\mathbb{D}_R)$ given by:
\begin{eqnarray*}\label{algNa}
	A^{\dagger} := z,\qquad A:=\partial_{\mathcal {R}(p,q)}, \qquad N:= z\partial_z,
\end{eqnarray*} 
where $\partial_z:=\frac{\partial}{\partial z}$ is the  derivative on $\mathbb{C}.$

The super multibracket of order $n$ is defined as \cite{HW}:
\begin{small}
	\begin{eqnarray}\label{smb}
	\Big[A_1,A_2,\cdots,A_n\Big]:=\epsilon^{i_1i_2\cdots i_n}_{12\cdots n}\big(-1\big)^{\sum_{k=1}^{n-1}|A_k|(\sum_{l=k+1,i_l< i_k}^{n}|A_{i_l}|)}\,A_{i_1}\,A_{i_2}\cdots A_{i_n},
	\end{eqnarray}
\end{small}
where the symbol $|A|$ is to be understood as the parity of $A$ and  $\epsilon^{i_1 \cdots i_n}_{1 \cdots n}$ is the L\'evi-Civit\'a symbol defined by:
\begin{eqnarray}\label{LCs}
\epsilon^{j_1 \cdots j_p}_{i_1 \cdots i_p}:= det\left( \begin{array} {ccc}
\delta^{j_1}_{i_1} &\cdots&  \delta^{j_1}_{i_p}   \\ 
\vdots && \vdots \\
\delta^{j_p}_{i_1} & \cdots& \delta^{j_p}_{i_p}
\end {array} \right) .
\end{eqnarray}
Moreover, the $q-$ deformed  generalized Jacobi identity is given by \cite{AP,AI}:
\begin{eqnarray*}
	\epsilon^{i_{1},\ldots,i_{2n-1}}_{n_{1},\ldots,n_{2n-1}}\big[\big[l_{i_1},\ldots,l_{i_{2n-1}}\big]_{q},l_{i_{n+1}},\ldots, l_{i_{2n-1}}\big]_q=0.
\end{eqnarray*}
\section{Super $\mathcal{R}(p,q)-$ deformed Witt $n-$ algebra}\label{sec3}
In this section, we contruct the super Witt algebra and the super Witt $n-$ algebra from the $\mathcal{R}(p,q)-$ deformed quantum algebra.

Let $\mathcal{B}={\mathcal B}_0\oplus \mathcal{B}_1$ be  the super-commutative associative superalgebra such that ${\mathcal B}_0=\mathbb{C}\big[z,z^{-1}\big]$ and $\mathcal{B}_1=\theta\,\mathcal{B}_0,$ where $\theta$ is the Grassman variable with $\theta^2=0$ \cite{WYLWZ}:

We define the algebra endomorphism $\sigma$ on $\mathcal{B}$ as follows:
\begin{eqnarray}
\sigma(t^n):=\big(\phi(p,q)\big)^n\,t^n\quad\mbox{and}\quad \sigma(\theta):=\phi(p,q)\theta,
\end{eqnarray}
where  $\phi(p,q)$ is a function  depending on the parameters $p$ and $q$ such that $\phi(p,q)\longrightarrow 1$ as $(p,q)\longrightarrow (1,1).$

We define also the two linear maps by:
\begin{eqnarray*}
	\left \{
	\begin{array}{l}
		\partial_t(t^n):=[n]_{{\mathcal R}(p,q)}\,t^n\mbox{,}\quad \partial_t(\theta\,t^n):=[n]_{{\mathcal R}(p,q)}\,\theta\,t^n, \\
		\\
		\partial_{\theta}(t^n):=0\mbox{,}\quad \partial_{\theta}(\theta\,t^n):=\big(\phi(p,q)\big)^n\,t^n.
	\end{array}
	\right .
\end{eqnarray*}
%
\begin{lemma}
	The linear map $\Delta=\partial_{t}+\theta \partial_{\theta}$ on  ${\mathcal B}$ is an even $\sigma$-derivation. Then:
	\begin{eqnarray}\label{deltaxy}
	\,\Delta(x\,y)&=&\Delta(x)\,y+\sigma(x)\Delta(y),\nonumber\\
	\,\Delta(t^n)&=& [n]_{{\mathcal R}(p,q)}\,t^n\quad\mbox{and}\quad \Delta(\theta\,t^n)= \big([n]_{{\mathcal R}(p,q)} + \big(\phi(p,q)\big)^n\big)\,\theta\,t^n. 
	\end{eqnarray}
\end{lemma}
\begin{proof}
	By direct computation.\cqfd
\end{proof}
		Taking $\mathcal{R}_{x1}=(q-1)^{-1}(x-1)$ and $\phi(q)=q,$ we obtained the result given in \cite{WYLWZ}. 

The  super $\mathcal{R}(p,q)-$ deformed Witt algebra is generated by bosonic and fermionic operators $l^{{\mathcal{R}(p,q)}}_m=-t^m\,\Delta$ of parity $0$ and $G^{{\mathcal{R}(p,q)}}_m=-\theta\,t^m\,\Delta$ of parity $1.$
\begin{proposition}
	The operators $l^{{\mathcal{R}(p,q)}}_m$ and $G^{{\mathcal{R}(p,q)}}_m$ satisfy the following relations:
	\begin{eqnarray}
	\,\big[l^{{\mathcal{R}(p,q)}}_{m_1},l^{{\mathcal{R}(p,q)}}_{m_2}\big]_{\hat{x}, \hat{y}}&=&\big([m_1]_{\mathcal{R}(p,q)}-[m_2]_{\mathcal{R}(p,q)}\big)\,l^{{\mathcal{R}(p,q)}}_{m_1+m_2},\label{crochet1}\\
	\,\big[l^{{\mathcal{R}(p,q)}}_{m_1}, G^{{\mathcal{R}(p,q)}}_{m_2}\big]_{x,y}&=&\big([m_1]_{\mathcal{R}(p,q)}-[m_2+1]_{\mathcal{R}(p,q)}\big)\,G^{{\mathcal{R}(p,q)}}_{m_1+m_2}\label{crochet2},\\
	\,\big[G^{{\mathcal{R}(p,q)}}_{m_1},G^{{\mathcal{R}(p,q)}}_{m_2}\big]&=&0,\label{crochet3}
	\end{eqnarray}
	
	where\begin{eqnarray}\label{coefcom}
	\left \{
	\begin{array}{l}
	\hat{x}=\chi_{m_1m_2}(p,q)\mbox{,}\quad \hat{y}=(\phi(p,q))^{m_2-m_1}\,\chi_{m_1m_2}(p,q), \\
	\\
	x=\tau_{m_1m_2}(p,q)\mbox{,}\quad y=(\phi(p,q))^{1+m_2-m_1}\,\tau_{m_1m_2}(p,q), \\
	\\
	\chi_{m_1m_2}(p,q)={[m_1]_{\mathcal{R}(p,q)}-[m_2]_{\mathcal{R}(p,q)}\over (\phi(p,q))^{m_2-m_1}\,[m_1]_{{\mathcal R}(p,q)}-[m_2]_{{\mathcal R}(p,q)}}\\
	\\
	\tau_{m_1m_2}(p,q)={[m_1]_{\mathcal{R}(p,q)}-[m_2+1]_{\mathcal{R}(p,q)}\over (\phi(p,q))^{1+m_2-m_1}\,[m_1]_{\mathcal{R}(p,q)}-[m_2]_{\mathcal{R}(p,q)}-(\phi(p,q))^{m_2}}.
	\end{array}
	\right .
	\end{eqnarray}
\end{proposition}
\begin{proof}
	From the definition of the deformed commutators, we get: 
	\begin{eqnarray}\label{com1}
	\big[l^{{\mathcal{R}(p,q)}}_{m_1},l^{{\mathcal{R}(p,q)}}_{m_2}\big]_{\hat{x}, \hat{y}}=\hat{x}\label{key}\,l^{{\mathcal{R}(p,q)}}_{m_1}l^{{\mathcal{R}(p,q)}}_{m_2}- \hat{y}\,l^{{\mathcal{R}(p,q)}}_{m_2}l^{{\mathcal{R}(p,q)}}_{m_1}.
	\end{eqnarray}
	Thus,
	\begin{eqnarray*}
		\hat{x}\,l^{{\mathcal{R}(p,q)}}_{m_1}.l^{{\mathcal{R}(p,q)}}_{m_2}&=&-t^{m_1}\,\Delta(l^{{\mathcal{R}(p,q)}}_{m_2})\nonumber\\
		&=& -\hat{x}\,[m_2]_{\mathcal{R}(p,q)}\,l^{{\mathcal{R}(p,q)}}_{m_1+m_2}-\hat{x}\,(\phi(p,q))^{m_2}\,l^{{\mathcal{R}(p,q)}}_{m_1+m_2}\,\Delta.
	\end{eqnarray*}
	Similarly, we have:
	\begin{eqnarray*}
		\hat{y}\,l^{{\mathcal{R}(p,q)}}_{m_2}l^{{\mathcal{R}(p,q)}}_{m_1}&=&- \hat{y}\,[m_1]_{\mathcal{R}(p,q)}\,l^{{\mathcal{R}(p,q)}}_{m_1+m_2}-\hat{y}\,(\phi(p,q))^{m_1}\,l^{{\mathcal{R}(p,q)}}_{m_1+m_2}\,\Delta.
	\end{eqnarray*}
	Then, the relation (\ref{com1}) takes the following form:
	\begin{eqnarray*}
		\big[l^{{\mathcal{R}(p,q)}}_{m_1},l^{{\mathcal{R}(p,q)}}_{m_2}\big]_{\hat{x}, \hat{y}}&=&\big(\hat{y}\,[m_1]_{\mathcal{R}(p,q)}-\hat{x}\,[m_2]_{\mathcal{R}(p,q)}\big)\,l^{{\mathcal{R}(p,q)}}_{m_1+m_2}\nonumber\\ &+& \big(\hat{y}\,(\phi(p,q))^{m_1}-\hat{x}\,(\phi(p,q))^{m_2}\big)\,l^{{\mathcal{R}(p,q)}}_{m_1+m_2}\,\Delta.
	\end{eqnarray*}
	We need to get
	\begin{eqnarray*}
		\big[l^{{\mathcal{R}(p,q)}}_{m_1},l^{{\mathcal{R}(p,q)}}_{m_2}\big]_{\hat{x}, \hat{y}}&=&\big([m_1]_{\mathcal{R}(p,q)}-[m_2]_{\mathcal{R}(p,q)}\big)l^{{\mathcal{R}(p,q)}}_{m_1+m_2}.
	\end{eqnarray*}
	Thus, we obtain the system:
	
	\begin{eqnarray*}
		\left \{
		\begin{array}{l}
			\hat{y}\,[m_1]_{\mathcal{R}(p,q)}-\hat{x}\,[m_2]_{\mathcal{R}(p,q)}=[m_1]_{\mathcal{R}(p,q)}-[m_2]_{\mathcal{R}(p,q)}\\
			\\ \hat{y}\,(\phi(p,q))^{m_1}-\hat{x}\,(\phi(p,q))^{m_2}=0.
		\end{array}
		\right .
	\end{eqnarray*}
	and 
	\begin{eqnarray*}
		\hat{x}&=&{[m_1]_{\mathcal{R}(p,q)}-[m_2]_{\mathcal{R}(p,q)}\over (\phi(p,q))^{m_2-m_1}\,[m_1]_{\mathcal{R}(p,q)}-[m_2]_{\mathcal{R}(p,q)}}\nonumber\\
		&:=& \chi_{m_1m_2}(p,q).
	\end{eqnarray*}
	After computation, we get
	\begin{eqnarray*}
		\hat{y}=(\phi(p,q))^{m_2-m_1}\,\chi_{m_1m_2}(p,q).
	\end{eqnarray*}
	Moreover, 
	\begin{small}
		$$	xl^{{\mathcal{R}(p,q)}}_{m_1}\,G^{{\mathcal{R}(p,q)}}_{m_2}
		=-x\big([m_2]_{\mathcal{R}(p,q)}+(\phi(p,q))^{m_2}\big)G^{{\mathcal{R}(p,q)}}_{m_1+m_2}- x(\phi(p,q))^{m_2+1}G^{{\mathcal{R}(p,q)}}_{m_1+m_2}\Delta$$
	\end{small}
	and 
	\begin{eqnarray*}
		y\,G^{{\mathcal{R}(p,q)}}_{m_2}\,l^{{\mathcal{R}(p,q)}}_{m_1}
		&=& -y\,[m_1]_{\mathcal{R}(p,q)}\,G^{{\mathcal{R}(p,q)}}_{m_1+m_2}-y\,(\phi(p,q))^{m_1}\,G^{{\mathcal{R}(p,q)}}_{m_1+m_2}\,\Delta.
	\end{eqnarray*}
	Thus, we get
	\begin{small}
		\begin{eqnarray*}
			\big[l^{{\mathcal{R}(p,q)}}_{m_1},G^{{\mathcal{R}(p,q)}}_{m_2}\big]_{x, y}&=& \Big(y\,[m_1]_{\mathcal{R}(p,q)}-x\,\big([m_2]_{\mathcal{R}(p,q)}+(\phi(p,q))^{m_2}\big)\Big)G^{{\mathcal{R}(p,q)}}_{m_1+m_2}\nonumber\\
			&+&  \big(y\,(\phi(p,q))^{m_1}-x\,(\phi(p,q))^{m_2+1}\big)\,G^{{\mathcal{R}(p,q)}}_{m_1+m_2}\,\Delta
		\end{eqnarray*}
	\end{small}
	and
	\begin{eqnarray*}
		\left \{
		\begin{array}{l}
			y\,[m_1]_{\mathcal{R}(p,q)}-x\,\big([m_2]_{\mathcal{R}(p,q)}+(\phi(p,q))^{m_2}\big)=[m_1]_{\mathcal{R}(p,q)}-[m_2+1]_{\mathcal{R}(p,q)}\\
			\\ y\,(\phi(p,q))^{m_1}-x(\phi(p,q))^{m_2+1}=0.
		\end{array}
		\right .
	\end{eqnarray*}
	Solving the above system, we obtain:
	\begin{eqnarray*}
		x&=&{[m_1]_{\mathcal{R}(p,q)}-[m_2+1]_{\mathcal{R}(p,q)}\over (\phi(p,q))^{1+m_2-m_1}\,[m_1]_{\mathcal{R}(p,q)}-[m_2]_{\mathcal{R}(p,q)}-(\phi(p,q))^{m_2}}\nonumber\\
		&:=& \tau_{m_1m_2}(p,q)
	\end{eqnarray*}
	and
	\begin{eqnarray*}
		y=(\phi(p,q))^{1+m_2-m_1}\,\tau_{m_1m_2}(p,q).
	\end{eqnarray*} \cqfd
\end{proof}

Let us now construct the super $\mathcal{R}(p,q)-$ deformed Witt $n-$ algebra. We define the $\mathcal{R}(p,q)-$ deformed $n-$ bracket $(n\geq 3)$  as follows:
\begin{small}
	\begin{eqnarray}\label{rnb1}
	\big[l^{{\mathcal{R}(p,q)}}_{m_1},\cdots, l^{\mathcal{R}(p,q)}_{m_n}\big]&:=&\bigg(\frac{[-2\sum_{l=1}^{n}m_{l}]_{\mathcal{R}(p,q)}}{2[-\sum_{l=1}^{n}m_{l}]_{\mathcal{R}(p,q)}}\bigg)^{\alpha}\epsilon^{i_1i_2\cdots i_n}_{12\cdots n}\nonumber\\&\times&(\phi(p,q))^{\sum_{j=1}^{n}\big(\lfloor \frac{n}{2} \rfloor-j+1\big)m_{i_j}}l^{\mathcal{R}(p,q)}_{m_{i_1}} \ldots
	l^{{\mathcal{R}(p,q)}}_{m_{i_n}},
	\end{eqnarray}
\end{small}
where $\alpha=\frac{1+(-1)^n }{ 2},$  $\lfloor n \rfloor=Max\{m\in\mathbb{Z}\ m\leq n\}$ is the floor function.

Introducing the operator $l^{{\mathcal{R}(p,q)}}_m=-t^m\,\Delta$ into the relation (\ref{rnb1}),  the $\mathcal{R}(p,q)-$ deformed  $n-$ bracket can be reduced in the simpler form as follows:
\begin{small}
	\begin{eqnarray*}
		\big[l^{{\mathcal{R}(p,q)}}_{m_1},l^{{\mathcal{R}(p,q)}}_{m_2},\ldots, l^{{\mathcal{R}(p,q)}}_{m_n}\big]&=&\frac{\big(q-p\big)^{n-1\choose 2}}{ (\phi(p,q))^{\lfloor {n-1\over 2}\rfloor\sum_{l=1}^{n}m_l}} \Bigg({[-2\sum_{l=1}^{n}m_l]_{{\mathcal R}(p,q)}\over 2[-\sum_{l=1}^{n}m_l]_{{\mathcal R}(p,q)}}\Bigg)^{\alpha}\nonumber\\&\times&\prod_{1\leq i < j\leq n}\Big([m_i]_{{\mathcal R}(p,q)}-[m_j]_{{\mathcal R}(p,q)}\Big)l_{\sum_{l=1}^{n}m_l}.
	\end{eqnarray*}
\end{small}
Now, we investigate the  super $\mathcal{R}(p,q)-$ deformed Witt $n-$ algebra. 

From the super multibracket of order $n$ (\ref{smb}), we define another $\mathcal{R}(p,q)-$ deformed $n-$ bracket as follows:
\begin{small}
	\begin{eqnarray}\label{rnb2}
	\big[l^{{\mathcal{R}(p,q)}}_{m_1},l^{{\mathcal{R}(p,q)}}_{m_2},\cdots, G^{{\mathcal{R}(p,q)}}_{m_n}\big]:&=&\bigg({[-2\sum_{l=1}^{n}m_l-1]_{\mathcal{R}(p,q)}\over 2[-\sum_{l=1}^{n}m_l-1]_{\mathcal{R}(p,q)}}\bigg)^{\alpha}\sum_{j=0}^{n-1}(-1)^{n-1+j}\epsilon^{i_1\ldots i_{n-1}}_{12\cdots n-1}\nonumber\\
	&\times&(\phi(p,q))^{\beta}l^{{\mathcal{R}(p,q)}}_{m_{i_1}}\ldots l^{{\mathcal{R}(p,q)}}_{m_{i_j}}
	G^{{\mathcal{R}(p,q)}}_{m_{n}}l^{{\mathcal{R}(p,q)}}_{m_{i_{j+1}}}\cdots l^{{\mathcal{R}(p,q)}}_{m_{i_{n-1}}},\qquad
	\end{eqnarray}
\end{small}
where $\beta=\sum_{k=1}^{j}\big(\lfloor {n\over 2} \rfloor-k+1\big)m_{i_k}+\big(\lfloor  {n\over 2} \rfloor-1\big)\big(m_n+1\big)+\sum_{k=j+1}^{n-1}\big(\lfloor {n\over 2} \rfloor -k\big)m_{i_k}.$

Using  the bosonic and fermionic operators, the ${\mathcal R}(p,q)-$ deformed  $n-$ bracket (\ref{rnb2}) can be rewritten as:
\begin{small}
	\begin{eqnarray*}
		\big[l^{{\mathcal{R}(p,q)}}_{m_1},l^{{\mathcal{R}(p,q)}}_{m_2},\cdots,G^{{\mathcal{R}(p,q)}}_{m_n}\big]&=& {\big(q-p\big)^{{n-1\choose 2}}\over \big(\phi(p,q)\big)^{\lfloor {n-1\over 2}\rfloor \sum_{l=1}^{n}m_l+1}}\bigg({[-2\sum_{l=1}^{n}m_l-1]_{{\mathcal R}(p,q)}\over 2[\sum_{l=1}^{n}m_l-1]_{{\mathcal R}(p,q)}}\bigg)^{\alpha}\nonumber\\&\times&
		\prod_{1\leq i<j\leq n-1} \big([m_i]_{{\mathcal R}(p,q)}- [m_j]_{{\mathcal R}(p,q)}\big)\nonumber\\&\times&\prod_{i=1}^{n-1} \big([m_i]_{{\mathcal R}(p,q)}- [m_n+1]_{{\mathcal R}(p,q)}\big)\,G^{{\mathcal{R}(p,q)}}_{\sum_{l=1}^{n}m_l}.
	\end{eqnarray*}
\end{small}
\begin{proposition}
	The super ${\mathcal R}(p,q)-$ deformed Witt  $n-$ algebras is generated by the operators $l^{{\mathcal R}(p,q)}_m$ and $G^{{\mathcal R}(p,q)}_m$ satisfying the following commutation  relations:
	\begin{small}
		\begin{eqnarray}\label{rcom1}
		\big[l^{{\mathcal{R}(p,q)}}_{m_1},l^{{\mathcal{R}(p,q)}}_{m_2},\ldots, l^{{\mathcal{R}(p,q)}}_{m_n}\big]&=&{\big(q-p\big)^{n-1\choose 2}\over (\phi(p,q))^{\lfloor {n-1\over 2}\rfloor\sum_{l=1}^{n}m_l}} \Big({[-2\sum_{l=1}^{n}m_l]_{{\mathcal R}(p,q)}\over 2[-\sum_{l=1}^{n}m_l]_{{\mathcal R}(p,q)}}\Big)^{\alpha}\nonumber\\&\times&\prod_{1\leq i < j\leq n}\Big([m_i]_{{\mathcal R}(p,q)}-[m_j]_{{\mathcal R}(p,q)}\Big)l_{\sum_{l=1}^{n}m_l}.
		\end{eqnarray}
	\end{small}
	and
	\begin{small}
		\begin{eqnarray}\label{rcom2}
		\big[l^{{\mathcal{R}(p,q)}}_{m_1},l^{{\mathcal{R}(p,q)}}_{m_2},\cdots,G^{{\mathcal{R}(p,q)}}_{m_n}\big]&=& {\big(q-p\big)^{{n-1\choose 2}}\over \big(\phi(p,q)\big)^{\lfloor {n-1\over 2}\rfloor \sum_{l=1}^{n}m_l+1}}\bigg({[-2\sum_{l=1}^{n}m_l-1]_{{\mathcal R}(p,q)}\over 2[\sum_{l=1}^{n}m_l-1]_{{\mathcal R}(p,q)}}\bigg)^{\alpha}\nonumber\\&\times&
		\prod_{1\leq i<j\leq n-1} \big([m_i]_{{\mathcal R}(p,q)}- [m_j]_{{\mathcal R}(p,q)}\big)\nonumber\\&\times&\prod_{i=1}^{n-1} \big([m_i]_{{\mathcal R}(p,q)}- [m_n+1]_{{\mathcal R}(p,q)}\big)\,G^{{\mathcal{R}(p,q)}}_{\sum_{l=1}^{n}m_l}
		\end{eqnarray}
	\end{small}
	and other anti-commutators are zeros.
\end{proposition}

		Taking $n=3$ in the relations (\ref{rcom1}) and (\ref{rcom2}), we obtain the super ${\mathcal R}(p,q)-$ deformed Witt $3-$ algebra:
		\begin{small}
			\begin{eqnarray*}
				\big[l^{\mathcal{R}(p,q)}_{m_1},l^{{\mathcal{R}(p,q)}}_{m_2}, l^{\mathcal{R}(p,q)}_{m_3}\big]&=&{\big(q-p\big)\over (\phi(p,q))^{ m_1+m_2+m_3}}\big([m_1]_{{\mathcal R}(p,q)}- [m_2]_{{\mathcal R}(p,q)}\big)\nonumber\\&
				\times& \big([m_1]_{{\mathcal R}(p,q)}- [m_3]_{{\mathcal R}(p,q)}\big)\big([m_2]_{{\mathcal R}(p,q)}- [m_3]_{{\mathcal R}(p,q)}\big)l_{m_1+m_2+m_3},
			\end{eqnarray*}
			\begin{eqnarray*}
				\big[l^{\mathcal{R}(p,q)}_{m_1},l^{\mathcal{R}(p,q)}_{m_2},G^{\mathcal{R}(p,q)}_{m_3}\big]&=& {\big(q-p\big)\big([m_1]_{{\mathcal R}(p,q)}- [m_2]_{{\mathcal R}(p,q)}\big)\over 2\big(\phi(p,q)\big)^{ m_1+m_2+m_3+3}}  \big([m_1]_{{\mathcal R}(p,q)}- [m_3+1]_{{\mathcal R}(p,q)}\big)\nonumber\\&
				\times&\big([m_2]_{{\mathcal R}(p,q)}- [m_3+1]_{{\mathcal R}(p,q)}\big)G_{m_1+m_2+m_3}
			\end{eqnarray*}
		\end{small}
		and other anti-commutators are zeros.

Now, we investigate the Virasoro $2n-$ algebra in the framework of the $\mathcal{R}(p,q)-$ deformed quantum algebras. 
The Virasoro algebra
$${\mathcal V}ir=\bigoplus_{n\in{\mathbb Z}}{\mathbb K}L_n \oplus {\mathbb K}\,C$$ is the Lie algebra which satisfies the commutation relations\cite{IK}: 
\begin{eqnarray*}\label{va}
	\left[L_{m}, L_{n} \right]=(m-n)L_{n+m} + \frac{1}{12}m(m-1)(m+1)\delta_{m+n,o}\,C,
\end{eqnarray*}
\begin{eqnarray*}
	\left[\mathcal{V}ir, C\right]=\{0\},
\end{eqnarray*}
where $\delta_{i,j}$ denotes the Kronecker delta and $C$ the central charge.

The $\mathcal{R}(p,q)-$ deformed operators $L_n$ defined as:
\begin{eqnarray*}
	L_n:= -t^n\,\bar{D}_{{\mathcal R}(p,q)}
\end{eqnarray*} 
satisfy  the $\mathcal{R}(p,q)-$ deformed Witt  $n-$ algebra given by (\ref{rcom1}). From the skewsymmetry and the $\mathcal{R}(p,q)-$ deformed generalized Jacobi identity,
we have:
\begin{lemma} 
	The ${\mathcal R}(p,q)-$ deformed Virasoro $2n-$ algebra is generated  by the following relation:
	\begin{eqnarray}\label{V2nalg}
	\big[L_{m_1},\cdots,L_{m_{2n}}\big] =g_{{\mathcal R}(p,q)}(m_1,\cdots,m_{2n}) + C_{{\mathcal R}(p,q)}(m_1,\cdots,m_{2n}),
	\end{eqnarray}
	where
	\begin{small}
		\begin{eqnarray}\label{gv}
		g_{{\mathcal R}(p,q)}(m_1,\cdots,m_{2n})&=&{(q-p)^{{2n-1\choose 2}}\over \big(\phi(p,q)\big)^{(n-1) \sum_{l=1}^{2n}m_l}}\bigg({[-2\sum_{l=1}^{2n}m_l]_{{\mathcal R}(p,q)}\over 2[-\sum_{l=1}^{2n}m_l]_{{\mathcal R}(p,q)}}\bigg)\nonumber\\
		&\times&\prod_{1\leq i< j\leq 2n}\Big([m_i]_{{\mathcal R}(p,q)}-[m_j]_{{\mathcal R}(p,q)}\Big)L_{\sum_{l=1}^{2n}m_l}
		\end{eqnarray}
	\end{small}
	and \begin{eqnarray}\label{cv}
	C_{{\mathcal R}(p,q)}(m_1,\cdots,m_{2n})&=&{c(p,q)\epsilon^{i_1\cdots i_{2n}}_{1\cdots 2n}\over 6\times  2^n\times n!}\prod_{l=1}^{n}{[m_{i_{2l-1}}-1]_{{\mathcal R}(p,q)}\over \big(\phi(p,q)\big)^{m_{2l-1}}}{[m_{2l-1}]_{{\mathcal R}(p,q)}\over [2m_{2l-1}]_{{\mathcal R}(p,q)}}\nonumber\\&\times& [m_{i_{2l-1}}]_{{\mathcal R}(p,q)}[m_{i_{2l-1}}+1]_{{\mathcal R}(p,q)}
	\delta_{m_{i_{2l-1}}+ m_{i_{2l}},0}
	\end{eqnarray}	
	is the $\mathcal{R}(p,q)-$ deformed central extension.
\end{lemma}
\begin{example}
	Some examples are given for $n=2$ and $n=3.$
	\begin{enumerate}
		\item [(a)]
		Taking $n=2$ in the realtions (\ref{V2nalg}), (\ref{gv}), and (\ref{cv}), we obtain the ${\mathcal R}(p,q)-$ deformed Virasoro $4-$ algebra:
		\begin{small}
			\begin{eqnarray*}
				\big[L_{m_1},L_{m_2},L_{m_3},L_{m_{4}}\big]_{{\mathcal R}(p,q)} =g_{{\mathcal R}(p,q)}(m_1,m_2,m_3,m_{4})+C_{{\mathcal R}(p,q)}(m_1,\cdots,m_{4}),
			\end{eqnarray*}
			where
			\begin{eqnarray*}
				g_{{\mathcal R}(p,q)}(m_1,m_2,m_3,m_4) &=& {(q-p)^{3}\over \big(\phi(p,q)\big)^{ m_1+m_2+m_3+m_4}}\bigg({[-2\sum_{l=1}^{4}m_l]_{{\mathcal R}(p,q)}\over 2[-\sum_{l=1}^{4}m_l]_{{\mathcal R}(p,q)}}\bigg)\nonumber\\&\times&\prod_{1\leq i < j\leq 4}\Big([m_i]_{{\mathcal R}(p,q)}-[m_j]_{{\mathcal R}(p,q)}\Big)L_{\sum_{l=1}^{4}m_l}
			\end{eqnarray*}
			and
			\begin{eqnarray*}
				C_{{\mathcal R}(p,q)}(m_1,\cdots,m_{4})&=&{c(p,q)\epsilon^{i_1\cdots i_{4}}_{1\cdots 4}\over 48}\prod_{l=1}^{2}\big(\phi(p,q)\big)^{-m_{2l-1}}{[m_{2l-1}]_{{\mathcal R}(p,q)}\over [2m_{2l-1}]_{{\mathcal R}(p,q)}}\nonumber\\&\times& [m_{i_{2l-1}}-1]_{{\mathcal R}(p,q)}[m_{i_{2l-1}}]_{{\mathcal R}(p,q)}[m_{i_{2l-1}}+1]_{{\mathcal R}(p,q)}
				\delta_{m_{i_{2l-1}}+ m_{i_{2l}},0.}
			\end{eqnarray*}
		\end{small}
		\item [(b)]The ${\mathcal R}(p,q)-$ deformed Virasoro $6-$ algebra is deduced from the generalization by taking $n=3:$
		\begin{eqnarray*}
			\big[L_{m_1},\cdots,L_{m_{6}}\big]_{{\mathcal R}(p,q)} =g_{{\mathcal R}(p,q)}(m_1,\cdots,m_{6})+C_{{\mathcal R}(p,q)}(m_1,\cdots,m_{6}),
		\end{eqnarray*}
		where
		\begin{small}
			\begin{eqnarray*}
				g_{{\mathcal R}(p,q)}(m_1,\cdots,m_{6})&=&{(q-p)^{10}\over \big(\phi(p,q)\big)^{2 \sum_{l=1}^{6}m_l}}\bigg({[-2\sum_{l=1}^{6}m_l]_{{\mathcal R}(p,q)}\over 2[-\sum_{l=1}^{6}m_l]_{{\mathcal R}(p,q)}}\bigg)\nonumber\\&\times&\prod_{1\leq i< j\leq 6}\Big([m_i]_{{\mathcal R}(p,q)}-[m_j]_{{\mathcal R}(p,q)}\Big)L_{\sum_{l=1}^{6}m_l}
			\end{eqnarray*}
			and
			\begin{eqnarray*}
				C_{{\mathcal R}(p,q)}(m_1,\cdots,m_{6})&=&{c(p,q)\epsilon^{i_1\cdots i_6}_{1\cdots 6}\over 288}\prod_{l=1}^{3}\big(\phi(p,q)\big)^{-m_{2l-1}}{[m_{2l-1}]_{{\mathcal R}(p,q)}\over [2m_{2l-1}]_{{\mathcal R}(p,q)}}\nonumber\\&\times& [m_{i_{2l-1}}-1]_{{\mathcal R}(p,q)}[m_{i_{2l-1}}]_{{\mathcal R}(p,q)}[m_{i_{2l-1}}+1]_{{\mathcal R}(p,q)}
				\delta_{m_{i_{2l-1}}+ m_{i_{2l}},0.}
			\end{eqnarray*}
		\end{small}
	\end{enumerate}
\end{example}

\section{Super$\mathcal{R}(p,q)-$ deformed  Virasoro $n-$ algebra}
In this section, we determine the super ${\mathcal R}(p,q)-$ deformed Jacobi identity.
 Furthermore,  we discuss the super $\mathcal{R}(p,q)-$ deformed  Virasoro algebra and derive the super $\mathcal{R}(p,q)-$ deformed Virasoro $n-$ algebra ($n$ even).
\begin{lemma}
	The ${\mathcal R}(p,q)-$ deformed superalgebra (\ref{crochet1}),(\ref{crochet2}), and (\ref{crochet3}) satisfies the  super ${\mathcal R}(p,q)-$ deformed Jacobi identity:
	\begin{eqnarray}
	\sum_{(i,j,l)\in\mathcal{C}(n,m,k)}\,(-1)^{|A_i||A_l|}\Big[\rho(A_i),\Big[A_j,A_l\Big]_{{\mathcal R}(p,q)}\Big]_{{\mathcal R}(p,q)}=0,
	\end{eqnarray} 
	where  $\rho(l^{{\mathcal R}(p,q)}_{m})={[2\,m]_{\mathcal{R}(p,q)}\over [m]_{\mathcal{R}(p,q)}}l^{{\mathcal R}(p,q)}_{m},$ $\rho(G^{{\mathcal R}(p,q)}_{m})={[2(m+1)]_{\mathcal{R}(p,q)}\over [m+1]_{\mathcal{R}(p,q)}}G^{{\mathcal R}(p,q)}_{m} $ and $\mathcal{C}(n,m,k)$ denotes the cyclic permutation of $(n,m,k)$. 
\end{lemma}
\begin{proof}
Taking respectively, $A_i=l^{\mathcal{R}(p,q)}_n,$ $A_j=l^{\mathcal{R}(p,q)}_m,$ $A_l=l^{\mathcal{R}(p,q)}_k,$ and by computation, 
the result follows. \cqfd
	\end{proof}

The  super $\mathcal{R}(p,q)-$ deformed Virasoro algebra is generated by bosonic and fermionic operators $\bar{l}^{{\mathcal{R}(p,q)}}_m=-t^m\,\Delta$ of parity $0$ and $\bar{G}^{{\mathcal{R}(p,q)}}_m=-\theta\,t^m\,\Delta$ of parity $1.$
\begin{proposition}
	The operators $\bar{l}^{{\mathcal{R}(p,q)}}$ and $\bar{G}^{{\mathcal{R}(p,q)}}_m$ satisfy  the following commutation relations:
	\begin{small}
		\begin{eqnarray}\label{gsva1}
		\big[\bar{l}^{{\mathcal{R}(p,q)}}_{m_1},\bar{l}^{{\mathcal{R}(p,q)}}_{m_2}\big]_{\hat{x}, \hat{y}}=\big([m_1]_{\mathcal{R}(p,q)}-[m_2]_{\mathcal{R}(p,q)}\big)\bar{l}^{\mathcal{R}(p,q)}_{m_1+m_2}+ C_{\mathcal{R}(p,q)}(m_1)\delta_{m_1+m_2,0}, \end{eqnarray}
		and
		\begin{eqnarray}\label{gsva2}
		\big[\bar{l}^{\mathcal{R}(p,q)}_{m_1}, \bar{G}^{{\mathcal{R}(p,q)}}_{m_2}\big]_{x,y}=\big([m_1]_{\mathcal{R}(p,q)}-[m_2+1]_{\mathcal{R}(p,q)}\big)\bar{G}^{{\mathcal{R}(p,q)}}_{m_1+m_2} + C_{\mathcal{R}(p,q)}(m_1)\delta_{m_1+m_2+1,0}, 
		\end{eqnarray}
	\end{small}
	where 
	$\hat{x},$ $\hat{y},$ $x,$ $y$ are given by the relation (\ref{coefcom}), \begin{eqnarray*}
		C_{\mathcal{R}(p,q)}(m_1)={c(p,q)(\phi(p,q))^{m_1}\,[m_1]_{{\mathcal R}(p,q)}\over 6[2m_1]_{{\mathcal R}(p,q)}} [m_1+1]_{{\mathcal R}(p,q)}[m_1]_{\mathcal{R}(p,q)}[m_1-1]_{{\mathcal R}(p,q)}
	\end{eqnarray*}
	is the $\mathcal{R}(p,q)-$ deformed central extension and other anti-commutators are zeros.
\end{proposition}
		Note that, the  super $q-$ deformed Virasoro algebra proposed by {\bf Ammar et {\it al}} \cite{AMS} can be recovered by taking   ${\mathcal R}(x,1)=(q-1)^{-1}(x-1).$

Following  the same procedure used to construct the ${\mathcal R}(p,q)-$ deformed Virasoro $2n-$ algebra (\ref{V2nalg}), we can also derive  the  super $\mathcal{R}(p,q)-$ deformed Virasoro $2n-$ algebra. It's generated by the  bosonic and fermionic operators $\bar{L}^{\mathcal{R}(p,q)}_m=-t^m\,\Delta$ of parity $0$ and $\bar{G}^{\mathcal{R}(p,q)}_m=-\theta\,t^m\,\Delta$ of parity $1$ satisfying the following relations: 
\begin{small}
	\begin{eqnarray*}
		\big[\bar{L}^{\mathcal{R}(p,q)}_{m_1},\cdots,\bar{L}^{\mathcal{R}(p,q)}_{m_{2n}}\big] =g_{{\mathcal R}(p,q)}(m_1,\cdots,m_{2n}) + C_{{\mathcal R}(p,q)}(m_1,\cdots,m_{2n}),
	\end{eqnarray*}
	\begin{eqnarray*}\label{sV2na}
		\big[\bar{L}^{\mathcal{R}(p,q)}_{m_1},\bar{L}^{\mathcal{R}(p,q)}_{m_2},\cdots, \bar{G}^{\mathcal{R}(p,q)}_{m_{2n}}\big]_{{\mathcal R}(p,q)}= f_{{\mathcal R}(p,q)}(m_1,m_2,\cdots m_{2n})+ {\mathcal CS}_{{\mathcal R}(p,q)}(m_1,\cdots m_{2n}),
	\end{eqnarray*}
\end{small}
where $g_{{\mathcal R}(p,q)}(m_1,\cdots,m_{2n})$ and $C_{{\mathcal R}(p,q)}(m_1,\cdots,m_{2n})$ are given by the relations (\ref{gv}), (\ref{cv}), 
\begin{small}
	\begin{eqnarray*}
		f_{{\mathcal R}(p,q)}(m_1,m_2,\cdots m_{2n})&=&{\big(q-p\big)^{{2n-1\choose 2}}\over \big(\phi(p,q)\big)^{-(n-1) \sum_{l=1}^{2n}m_l+1}}\bigg({[-2\sum_{l=1}^{2n}m_l-1]_{{\mathcal R}(p,q)}\over 2[\sum_{l=1}^{2n}m_l-1]_{{\mathcal R}(p,q)}}\bigg)\nonumber\\&\times&\prod_{1\leq i<j\leq 2n-1} \big([m_i]_{{\mathcal R}(p,q)}- [m_j]_{{\mathcal R}(p,q)}\big)\nonumber\\&\times&
		\prod_{i=1}^{2n-1} \big([m_i]_{{\mathcal R}(p,q)}- [m_{2n}+1]_{{\mathcal R}(p,q)}\big)G_{\sum_{l=1}^{2n}m_l},
	\end{eqnarray*}
\end{small}
\begin{small}
	\begin{eqnarray*}
		{\mathcal CS}_{\mathcal{R}(p,q)}(m_1,m_2,\cdots m_{2n})&=&\sum_{k=1}^{2n-1}{(-1)^{k+1}c(p,q)(\phi(p,q))^{-m_k}\over 6\times 2^{n-1}(n-1)!}{[m_{k}]_{\mathcal{R}(p,q)}\over [2m_{k}]_{{\mathcal R}(p,q)}}\nonumber\\&\times&[m_k+1]_{{\mathcal R}(p,q)}[m_k]_{{\mathcal R}(p,q)}[m_k-1]_{{\mathcal R}(p,q)}\delta_{m_k+m_{2n}+1,0}\nonumber\\&\times&\epsilon^{i_1\cdots i_{2n-2}}_{j_1\cdots j_{2n-2}}\prod_{s=1}^{n-1}{(\phi(p,q))^{-i_{2s-1}}[i_{2s-1}]_{{\mathcal R}(p,q)}\over [2\,i_{2s-1}]_{{\mathcal R}(p,q)}}\nonumber\\&\times&[i_{2s-1}+1]_{{\mathcal R}(p,q)}\,[i_{2s-1}]_{{\mathcal R}(p,q)}\,[i_{2s-1}-1]_{{\mathcal R}(p,q)}\delta_{i_{2s-1}+i_{2s},0},
	\end{eqnarray*}
\end{small}
with $\{j_1,\cdots, j_{2n-2}\}=\{1,\cdots,\hat{k},\cdots,2n-1\}$ and other anti-commutators are zeros.
	\section{A toy model for the  super $\mathcal{R}(p,q)-$ Virasoro constraints }
In this section, we construct another   super Witt $n-$ algebra from the $\mathcal{R}(p,q)-$ deformed quantum algebra. We use the super $\mathcal{R}(p,q)-$ Virasoro constraints to study a toy model. 

We consider  the operators defined by:
\begin{eqnarray}\label{to}
{\mathcal T}^{{\mathcal R}(p^{a},q^{a})}_m&:=&\Delta\,z^{m}\\
\mathbb{T}^{{\mathcal R}(p^{a},q^{a})}_m&:=&-\theta\,\Delta\,z^{m}\label{go}.
\end{eqnarray}
The operators (\ref{to}) and (\ref{go}) can be rewritten as:
\begin{eqnarray*}
	\mathcal{T}^{{\mathcal R}(p^{a},q^{a})}_m&=&-[m]_{{\mathcal R}(p^{a},q^{a})}\,z^{m}\\
	\mathbb{T}^{{\mathcal R}(p^{a},q^{a})}_m&=&-\theta\,[m]_{{\mathcal R}(p^{a},q^{a})}\,z^{m}.
\end{eqnarray*}
The $\mathcal{R}(p,q)-$ deformed numbers (\ref{rpqnumber}) can be rewritten as \cite{HMM}:
\begin{eqnarray*}
	[n]_{\mathcal{R}(p,q)}=\frac{\tau^{n}_1-\tau^{n}_2}{\tau_{1}-\tau_{2}}, \quad \tau_{1}\neq \tau_{2},
\end{eqnarray*}
where $\tau_{i}, i\in\{1,2\}$ are the functions depending on the deformation parameters $p$ and $q.$ For illustration, we have some particular cases \cite{HMM}:
\begin{enumerate}
	\item [(i)]$q-$ Arick-Coon-Kuryskin  deformation \cite{AC,K}
	\begin{eqnarray*}
		\tau_1=1, \quad \tau_2=q \quad \mbox{and} \quad
		[n]_q ={1-q^n\over 1-q};
	\end{eqnarray*}
	\item [(ii)]$(p,q)-$ Jagannathan-Srinivasa deformation \cite{JS}
	\begin{eqnarray*}
		\tau_1=p, \quad \tau_2=q \quad \mbox{and} \quad
		[n]_{p,q} ={p^n-q^n\over p-q}.
	\end{eqnarray*}
\end{enumerate}
\begin{lemma}
	The following products hold.
	\begin{small}
		\begin{eqnarray}\label{rpqprod1}
		\mathcal{T}^{\mathcal{R}(p^{a},q^{a})}_m\,. \mathcal{T}^{\mathcal{R}(p^{b},q^{b})}_n&=&-{\big(\tau^{a+b}_1-\tau^{a+b}_2\big)\tau^{-m\,b}_1\over \big(\tau^{a}_1-\tau^{a}_2\big)\big(\tau^{b}_1-\tau^{b}_2\big)}{\mathcal T}^{{\mathcal R}(p^{a+b},q^{a+b})}_{m+n}\nonumber\\&+&{\tau^{-n\,b}_2\over \tau^{b}_1-\tau^{b}_2}{\mathcal T}^{{\mathcal R}(p^{a},q^{a})}_{m+n} + {\tau^{(m+n)a}_2\tau^{-m\,b}_1\over \tau^{a}_1-\tau^{a}_2}{\mathcal T}^{{\mathcal R}(p^{b},q^{b})}_{m+n}
		\end{eqnarray}
		and
		\begin{eqnarray}\label{rpqprod2}
		\mathcal{T}^{\mathcal{R}(p^{a},q^{a})}_m\,. \mathbb{T}^{\mathcal{R}(p^{b},q^{b})}_n&=&-{\big(\tau^{a+b}_1-\tau^{a+b}_2\big)\tau^{-(m+1)b}_1\over \big(\tau^{a}_1-\tau^{a}_2\big)\big(\tau^{b}_1-\tau^{b}_2\big)}\mathbb{T}^{{\mathcal R}(p^{a+b},q^{a+b})}_{m+n+1}\nonumber\\&+&{\tau^{-n\,b}_2\over \tau^{b}_1-\tau^{b}_2}\mathbb{T}^{\mathcal{R}(p^{a},q^{a})}_{m+n+1}+ {\tau^{(m+n+1)\,a}_2\tau^{-(m+1)b}_1\over \tau^{a}_1-\tau^{a}_2}\mathbb{T}^{\mathcal{R}(p^{b},q^{b})}_{m+n+1} .
		\end{eqnarray}
	\end{small}
\end{lemma}
\begin{proposition}
	The operators (\ref{to}) and (\ref{go}) satisfy the following
	commutation relations:
	\begin{small}
		\begin{eqnarray}\label{scrto}
		\Big[\mathcal{T}^{\mathcal{R}(p^{a},q^{a})}_m, \mathcal{T}^{\mathcal{R}(p^{b},q^{b})}_n\Big]&=&{\big(\tau^{a+b}_1-\tau^{a+b}_2\big)\big(\tau^{-na}_1-\tau^{-mb}_1\big)\over \big(\tau^{a}_1-\tau^{a}_2\big)\big(\tau^{b}_1-\tau^{b}_2\big)}{\mathcal T}^{{\mathcal R}(p^{a+b},q^{a+b})}_{m+n}\nonumber\\ &-&{\tau^{(m+n)b}_2\big(\tau^{-na}_1-\tau^{-mb}_2\big)\over \tau^{b}_1-\tau^{b}_2}{\mathcal T}^{\mathcal {R}(p^{a},q^{a})}_{m+n}\nonumber\\&+& \frac{\tau^{(m+n)a}_2\big(\tau^{-mb}_1-\tau^{-na}_2\big)}{\tau^{a}_1-\tau^{a}_2}{\mathcal T}^{{\mathcal R}(p^{b},q^{b})}_{m+n},
		\end{eqnarray}
		\begin{eqnarray}\label{scrgo}
		\Big[{\mathcal T}^{{\mathcal R}(p^{a},q^{a})}_m, \mathbb{T}^{{\mathcal R}(p^{b},q^{b})}_n\Big]&=&{(\tau^{a+b}_1-\tau^{a+b}_2)(\tau^{-na}_1-\tau^{-mb+a}_1)\over (\tau^{a}_1-\tau^{a}_2)(\tau^{b}_1-\tau^{b}_2)}\mathbb{T}^{{\mathcal R}(p^{a+b},q^{a+b})}_{m+n}\nonumber\\ &+&{\tau^{b(m+n)}_2(\tau^{-bm}_2\tau^{a}_1-\tau^{-an}_1)\over \tau^{b}_1-\tau^{b}_2}\mathbb{T}^{{\mathcal R}(p^{a},q^{a})}_{m+n}\nonumber\\&+& {\tau^{a(m+n)}_2(\tau^{-mb}_1\tau^{a}_2-\tau^{-an}_2)\over \tau^{a}_1-\tau^{a}_2}{\mathbb T}^{{\mathcal R}(p^{b},q^{b})}_{m+n}+ f(m,n),
		\end{eqnarray}
	\end{small}
	where
	\begin{eqnarray*}
		f(m,n)&=&-{\big(\tau^{a+b}_1-\tau^{a+b}_2\tau^{-(m+1)b}_{1}\tau^{b}_2\big)}\mathbb{T}^{{\mathcal R}(p^{a+b},q^{a+b})}_{1}\nonumber\\ &+&{\tau^{(m+n)a}_2\tau^{n\,b}_2\over \tau^{b}_1-\tau^{b}_2}\mathbb{T}^{{\mathcal R}(p^{a},q^{a})}_{1}+ {\tau^{(m+n)(a+b)}_2\tau^{-(m+1)b}_1\tau^{a}_2\over \tau^{a}_1-\tau^{a}_2}\mathbb{T}^{{\mathcal R}(p^{b},q^{b})}_{1}
	\end{eqnarray*}
	and other anti-commutators are zeros.
\end{proposition}
Setting $a=b=1,$ we obtain:
\begin{small}
	\begin{eqnarray*}
		\Big[{\mathcal T}^{{\mathcal R}(p,q)}_m, {\mathcal T}^{{\mathcal R}(p,q)}_n\Big]&=&{(\tau^{-n}_1-\tau^{-m}_1)\over (\tau_1-\tau_2)}\,[2]_{{\mathcal R}(p,q)}{\mathcal T}^{{\mathcal R}(p^{2},q^{2})}_{m+n}\nonumber\\&-&{\tau^{m+n}_2\over \tau_1-\tau_2}\Big((\tau^{-n}_1-\tau^{-m}_2)-(\tau^{-m}_1-\tau^{-n}_2)\Big) {\mathcal T}^{{\mathcal R}(p,q)}_{m+n},
	\end{eqnarray*}
\end{small}
\begin{small}
	\begin{eqnarray*}
		\Big[{\mathcal T}^{{\mathcal R}(p,q)}_m, \mathbb{T}^{{\mathcal R}(p,q)}_n\Big]&=&{\big(\tau^{-n}_1-\tau^{-m+1}_1\big)\over \tau_1-\tau_2}\,[2]_{{\mathcal R}(p,q)}\mathbb{T}^{{\mathcal R}(p^{2},q^{2})}_{m+n}+ f(m,n)\nonumber\\ &+&{\tau^{m+n}_2\over \tau_1-\tau_2}\bigg(\big(\tau^{-m}_2\,\tau_1-\tau^{-n}_1\big)-\big(\tau^{-m}_1\tau_2-\tau^{-n}_2\big)\bigg)\mathbb{T}^{{\mathcal R}(p,q)}_{m+n},
	\end{eqnarray*}
where
\begin{eqnarray*}
	f(m,n)=-{\tau^{-m-1}_{1}\tau^{2(m+n)}_{2}\over \big(\tau_1-\tau_2\big)}\,[2]_{{\mathcal R}(p,q)}\mathbb{T}^{{\mathcal R}(p^{2},q^{2})}_{1} +{\tau^{m+n}_2\big(\tau^{n}_2+\tau^{m+n}_2\tau^{-m-1}_1\tau_2\big)\over \tau_1-\tau_2}\mathbb{T}^{{\mathcal R}(p,q)}_{1}
\end{eqnarray*}
\end{small}
and other anti-commutators are zeros.

We consider the $n-$ bracket defined by:
\begin{eqnarray*}
	\Big[{\mathcal T}^{{\mathcal R}(p^{a_1},q^{a_1})}_{m_1},\cdots,{\mathcal T}^{{\mathcal R}(p^{a_n},q^{a_n})}_{m_n}
	\Big]:=\epsilon^{i_1 \cdots i_n}_{1 \cdots n}\,{\mathcal T}^{{\mathcal R}(p^{a_{i_1}},q^{a_{i_1}})}_{m_{i_1}} \cdots {\mathcal T}^{{\mathcal R}(p^{a_{i_n}},q^{a_{i_n}})}_{m_{i_n}},
\end{eqnarray*}
where $\epsilon^{i_1 \cdots i_n}_{1 \cdots n}$ is the L\'evi-Civit\'a symbol defined by (\ref{LCs}).
Our study is focused in the case with the same ${\mathcal R}(p^{a},q^{a})$ leads to
\begin{eqnarray*}
	\Big[{\mathcal T}^{{\mathcal R}(p^{a},q^{a})}_{m_1},\cdots,{\mathcal T}^{{\mathcal R}(p^{a},q^{a})}_{m_n}
	\Big]=\epsilon^{1\cdots n}_{1\cdots n}\,{\mathcal T}^{{\mathcal R}(p^{a},q^{a})}_{m_{1}}\cdots {\mathcal T}^{{\mathcal R}(p^{a},q^{a})}_{m_{n}}.
\end{eqnarray*}
Putting $a=b$ in the relation (\ref{scrto}), we obtain:
\begin{small}
	\begin{eqnarray*}\label{crtob}
		\Big[{\mathcal T}^{{\mathcal R}(p^{a},q^{a})}_m, {\mathcal T}^{{\mathcal R}(p^{a},q^{a})}_n\Big]&=&{\big(\tau^{-na}_1-\tau^{-ma}_1\big)\over \big(\tau^{a}_1-\tau^{a}_2\big)}\,[2]_{{\mathcal R}(p^{a},q^{a})}{\mathcal T}^{{\mathcal R}(p^{2a},q^{2a})}_{m+n}\nonumber\\&-&{\tau^{(m+n)a}_2\over \tau^{a}_1-\tau^{a}_2}\Big(\big(\tau^{-na}_1-\tau^{-ma}_1\big)+\big(\tau^{-na}_2-\tau^{-ma}_2\big)\Big) {\mathcal T}^{{\mathcal R}(p^{a},q^{a})}_{m+n}.
	\end{eqnarray*}
\end{small}
The $n-$ bracket takes the following form:
\begin{small}
	\begin{eqnarray*}\label{crna}
		\Big[\mathcal{T}^{{\mathcal R}(p^{a},q^{a})}_{m_1},\cdots, {\mathcal T}^{{\mathcal R}(p^{a},q^{a})}_{m_n}\Big]&=&{(-1)^{n+1}\over \big(\tau^{a}_1-\tau^{a}_2\big)^{n-1}}\Big( M^n_{a}[n]_{{\mathcal R}(p^{a},q^{a})}{\mathcal T}^{{\mathcal R}(p^{n\,a},q^{n\,a})}_{m_1+\cdots+m_n}\nonumber\\ &-& {[n-1]_{{\mathcal R}(p^{a},q^{a})}\over  \tau^{-a\big(\sum_{l=1}^{n}m_l\big)}_2}\big(M^n_{a}+ C^n_{a}\big){\mathcal T}^{{\mathcal R}(p^{(n-1)a},q^{(n-1)a})}_{m_1+\cdots+m_n}\Big),
	\end{eqnarray*}
\end{small}
where 
\begin{small}
	\begin{eqnarray*}
		M^n_{a}&=& \tau^{-a(n-1)\sum_{s=1}^{n}m_s}_1\Big(\big(\tau^{a}_1-\tau^{a}_2\big)^{n\choose 2}\prod_{1\leq j < k \leq n}\Big([m_k]_{{\mathcal R}(p^{a},q^{a})}-[m_j]_{{\mathcal R}(p^{a},q^{a})}\Big)\nonumber\\&+&\prod_{1\leq j < k \leq n}\Big(\tau^{a\,m_k}_2-\tau^{a\,m_j}_2\Big)\Big)
	\end{eqnarray*}
	and 
	\begin{eqnarray*}
		C^{n}_{a}
		&=&\tau^{-a(n-1)\sum_{s=1}^{n}m_s}_2\Big(\big(\tau^{a}_1-\tau^{\alpha}_2\big)^{n\choose 2}\prod_{1\leq j < k \leq n}\Big([m_k]_{{\mathcal R}(p^{a},q^{a})}-[m_j]_{{\mathcal R}(p^{a},q^{a})}\Big)\nonumber\\&+&(-1)^{n-1}\prod_{1\leq j < k \leq n}\Big(\tau^{a\,m_k}_1-\tau^{a\,m_j}_1\Big)\Big).
	\end{eqnarray*}
\end{small}
From the super multibracket of order $n$ (\ref{smb}), we define the $\mathcal{R}(p,q)-$ deformed super $n-$ bracket as follows:
\begin{small}
	\begin{eqnarray*}\label{snbracket}
		\big[\mathcal{T}^{{\mathcal R}(p^{a},q^{a})}_{m_1},\mathcal{T}^{{\mathcal R}(p^{a},q^{a})}_{m_2},\cdots, \mathbb{T}^{{\mathcal R}(p^{a},q^{a})}_{m_n}\big]&:=&\sum_{j=0}^{n-1}(-1)^{n-1+j}\epsilon^{i_1\ldots i_{n-1}}_{12\cdots n-1}\mathcal{T}^{{\mathcal R}(p^{a},q^{a})}_{m_{i_1}}\cdots \mathcal{T}^{{\mathcal R}(p^{a},q^{a})}_{m_{i_j}}\nonumber\\&\times&
		\mathbb{T}^{{\mathcal R}(p^{a},q^{a})}_{m_{n}}\mathcal{T}^{{\mathcal R}(p^{a},q^{a})}_{m_{i_{j+1}}}\cdots \mathcal{T}^{{\mathcal R}(p^{a},q^{a})}_{m_{i_{n-1}}}.
	\end{eqnarray*}
\end{small}
From the relation (\ref{scrgo}) with $a=b,$ we obtain:
\begin{small}
	\begin{eqnarray*}
		\Big[{\mathcal T}^{{\mathcal R}(p^{a},q^{a})}_m, \mathbb{T}^{{\mathcal R}(p^{a},q^{a})}_n\Big]&=&{\big(\tau^{-an}_1-\tau^{-(m-1)a}_1\big)\over \big(\tau^{a}_1-\tau^{a}_2\big)}[2]_{{\mathcal R}(p^{a},q^{a})}\mathbb{T}^{{\mathcal R}(p^{2a},q^{2a})}_{m+n}+ f(m,n)\nonumber\\ &+&{\tau^{(m+n)a}_2\over \tau^{a}_1-\tau^{a}_2}\bigg((\tau^{-am}_2\tau^{a}_1-\tau^{-an}_1)+(\tau^{-am}_1\tau^{a}_2-\tau^{-an}_2)\bigg)\mathbb{T}^{{\mathcal R}(p^{a},q^{a})}_{m+n},
	\end{eqnarray*}
	where
	\begin{eqnarray*}
		f(m,n)=-{\tau^{-(m+1)a}_{1}\tau^{a(m+n)}_{2}\over \big(\tau^{a}_1-\tau^{a}_2\big)}\Big(\tau^{am}_2[2]_{{\mathcal R}(p^{a},q^{a})}\mathbb{T}^{{\mathcal R}(p^{2a},q^{2a})}_{1}-{[2(m+1)]_{\mathcal{R}(p^{a},q^{a})}\over [m+1]_{\mathcal{R}(p^{a},q^{a})}}\mathbb{T}^{\mathcal{R}(p^{a},q^{a})}_{1}\Big).
	\end{eqnarray*}
\end{small}
Thus, the super $n-$ bracket can be rewritten as follows:
\begin{small}
	\begin{eqnarray*}
		\Big[\mathcal{T}^{{\mathcal R}(p^{a},q^{a})}_{m_1},\cdots, \mathbb{T}^{{\mathcal R}(p^{a},q^{a})}_{m_n}\Big]&=&{(-1)^{n+1}\over \big(\tau^{a}_1-\tau^{a}_2\big)^{n-1}}\Big( A^n_{a}[n]_{{\mathcal R}(p^{a},q^{a})}\mathbb{T}^{{\mathcal R}(p^{n\,a},q^{n\,a})}_{m_1+\cdots+m_n}\nonumber\\ &-& {[n-1]_{{\mathcal R}(p^{a},q^{a})}\over  \tau^{-a\big(\sum_{l=1}^{n}m_l\big)}_2}\big(F^n_{a}+ S^n_{a}\big){\mathcal T}^{{\mathcal R}(p^{(n-1)a},q^{(n-1)a})}_{m_1+\cdots+m_n}\Big)\nonumber\\&+& f\big(m_{1},\cdots, m_{n}\big),
	\end{eqnarray*}
\end{small}
where 
\begin{small}
	\begin{eqnarray*}
		A^n_{a}&=& \tau^{-a(n-1)\sum_{s=1}^{n}(m_s-1)}_1\Big(\big(\tau^{a}_1-\tau^{a}_2\big)^{n\choose 2}\prod_{1\leq j < k \leq n}\Big([m_k-1]_{{\mathcal R}(p^{a},q^{a})}-[m_j]_{{\mathcal R}(p^{a},q^{a})}\Big)\nonumber\\&+&\prod_{1\leq j < k \leq n}\Big(\tau^{a(m_k-1)}_2-\tau^{a\,m_j}_2\Big)\Big),
	\end{eqnarray*}
	\begin{eqnarray*}
		F^n_{a}&=& \tau^{-a(n-1)\sum_{s=1}^{n}m_s}_1\Big(\big(\tau^{a}_1-\tau^{a}_2\big)^{n\choose 2}\prod_{1\leq j < k \leq n}\Big([m_k]_{{\mathcal R}(p^{a},q^{a})}-[m_j]_{{\mathcal R}(p^{a},q^{a})}\tau^{n\choose 2}_2\Big)\nonumber\\&+&\prod_{1\leq j < k \leq n}\Big(\tau^{a\,m_k}_2-\tau^{a\,m_j}_2\tau^{n\choose 2}_2\Big)\Big),
	\end{eqnarray*} 
	\begin{eqnarray*}
		S^{n}_{a}
		&=&\tau^{-a(n-1)\sum_{s=1}^{n}m_s}_2\Big(\big(\tau^{a}_1-\tau^{\alpha}_2\big)^{n\choose 2}\prod_{1\leq j < k \leq n}\Big([m_k]_{{\mathcal R}(p^{a},q^{a})}-[m_j]_{{\mathcal R}(p^{a},q^{a})}\tau^{n\choose 2}_1\Big)\nonumber\\&+&(-1)^{n-1}\prod_{1\leq j < k \leq n}\Big(\tau^{a\,m_k}_1-\tau^{a\,m_j}_1\tau^{n\choose 2}_1\Big)\Big)
	\end{eqnarray*}
	and 
	\begin{eqnarray*}
		f\big(m_{1},\ldots, m_{n}\big)&=&{(-1)^{n+1}\tau^{-(m+1)a}_{1}\tau^{a\sum_{l=1}^{n}m_l}_{2}\over \big(\tau^{a}_1-\tau^{a}_2\big)^{n-1}}\Big( \tau^{a\,m}_2[n]_{{\mathcal R}(p^{a},q^{a})}\mathbb{T}^{{\mathcal R}(p^{n\,a},q^{n\,a})}_{1}\nonumber\\ &-& {[2(m+1)]_{\mathcal{R}(p^{a},q^{a})}\over [m+1]_{\mathcal{R}(p^{a},q^{a})}}\mathbb{T}^{{\mathcal R}(p^{(n-1)a},q^{(n-1)a})}_{1}\Big).
	\end{eqnarray*}
\end{small}	

Let us consider the generating function with infinitely many	parameters
presented by  \cite{NZ}: $$Z^{toy}(t)=\int \, \,x^{\gamma}\,\exp\left(\displaystyle\sum_{s=0}^{\infty}{t_s\over s!}x^s\right)\,dx.$$
We assume that the following relation holds for the linear maps $\Delta$ given by the relation (\ref{deltaxy})
\begin{small}
	\begin{eqnarray*}
		\int_{\mathbb{R}}\,\Delta \,f(x)d\,x=0.
	\end{eqnarray*} 
\end{small}
Taking $f(x)=x^{m+\gamma}\,\exp\left(\displaystyle\sum_{s=0}^{\infty}{t_s\over s!}x^s\right),$ we have
\begin{eqnarray*}
	\int_{-\infty}^{+\infty}\Delta \,\left(x^{m+\gamma}\,\exp\left(\sum_{s=0}^{\infty}{t_s\over s!}x^s\right)\right)d\,x=0.
\end{eqnarray*}
We consider the following expression
\begin{eqnarray*}
	\exp\left(\displaystyle\sum_{s=0}^{\infty}{t_s\over s!}x^s\right)=\sum_{n=0}^{\infty}B_n(t_1,\cdots,t_n){x^n\over n!},
\end{eqnarray*}
where $B_n$ is the Bell polynomials. Then 
\begin{small}
	\begin{eqnarray*}
		\Delta \left(x^{m+\gamma}\,\exp\left(\displaystyle\sum_{s=0}^{\infty}{t_s\over s!}x^s\right)\right)
		&=&x^{m+\gamma}[m+\gamma]_{{\mathcal R}(p^{a},q^{a})}\exp\left(\displaystyle\sum_{s=0}^{\infty}{t_s\over s!}x^s\right)\nonumber\\ &+& 
		{\big(\phi(p,q)\big)^{m+\gamma}\over (\tau^{a}_1 -\tau^{a}_2)x^{-k-m}}\sum_{k=1}^{\infty}{B_k(t^{a}_1,\cdots,t^{a}_k)\over k!}\exp\left(\displaystyle\sum_{s=0}^{\infty}{t_s\over s!}x^{s+\gamma}\right),
	\end{eqnarray*}
	where $t^{a}_k=(\tau^{a\,k}_1-\tau^{a\,k}_2)t_k.$ Then,  from the relation 
	$$\mathcal{T}^{{\mathcal R}(p^{a},q^{a})}_m\,Z^{(toy)}(t)=0,\quad m\geq 0,$$
	the operator (\ref{to}) takes the following form:
	\begin{eqnarray*}
		\mathcal{T}^{{\mathcal R}(p^{a},q^{a})}_{m}=[m+\gamma]_{{\mathcal R}(p^{a},q^{a})}\,m!\,{\partial\over \partial t_m}+ {\big(\phi(p,q)\big)^{m+\gamma}\over \tau^{a}_1 - \tau^{a}_2}\sum_{k=1}^{\infty}{(k+m)!\over k!}B_k(t^{a}_1,\cdots,t^{a}_k){\partial\over \partial t_{k+m}}.
	\end{eqnarray*}
	Similarly, we obtain
	$$\mathbb{T}^{{\mathcal R}(p^{a},q^{a})}_m\,Z^{(toy)}(t)=0,\quad m\geq 0,$$
	and
	\begin{eqnarray*}
		\mathbb{T}^{{\mathcal R}(p^{a},q^{a})}_m=\theta\bigg([m+\gamma]_{{\mathcal R}(p^{a},q^{a})}\,m!\,{\partial\over \partial t_m} + {\big(\phi(p,q)\big)^{m+\gamma}\over \tau^{a}_1 - \tau^{a}_2}\sum_{k=1}^{\infty}{(k+m)!\over k!}B_k(t^{a}_1,\cdots,t^{a}_k){\partial\over \partial t_{k+m}}\bigg).
	\end{eqnarray*}
\end{small} 
Putting $\bar{m}=m+\gamma,\quad \bar{n}=n+\gamma,$ and by changing $n!\,{\partial\over \partial t_n}\longleftrightarrow x^n,$ we show directly that  the products $\mathcal{T}^{{\mathcal R}(p^{a},q^{a})}_{m}\,.\mathcal{T}^{{\mathcal R}(p^{b},q^{b})}_{n}$ and $\mathcal{T}^{{\mathcal R}(p^{a},q^{a})}_{m}\,.\mathbb{T}^{{\mathcal R}(p^{b},q^{b})}_{n}$ are respectively equivalent to (\ref{rpqprod1}) and (\ref{rpqprod2}).

\subsection{$q-$ deformed super Virasoro constraints}
The results obtained here can be deduced from the general formalism
by setting $\mathcal{R}(x,1)=(q-1)^{-1}(x-1).$ Then, the $q-$ deformed  operators given by:
\begin{eqnarray}\label{stoAC}
\,{\mathcal T}^{q^{a}}_m&=&\Delta\,z^{m}\\
\,\mathbb{T}^{q^{a}}_m&=&-\theta\,\Delta\,z^{m}\label{sgoAC}
\end{eqnarray}
satisfy the  products 
\begin{eqnarray}\label{ACprod1}
\mathcal{T}^{q^{a}}_m\,. \mathcal{T}^{q^{b}}_n=-{\big(q^{a+b}-1\big)\over \big(q^{a}-1\big)\big(q^{b}-1\big)}{\mathcal T}^{q^{a+b}}_{m+n}+{1\over q^{b}-1}{\mathcal T}^{q^{a}}_{m+n} + {q^{-m\,b}\over q^{a}-1}{\mathcal T}^{q^{b}}_{m+n}
\end{eqnarray}
and
\begin{small}
	\begin{eqnarray}\label{ACprod2}
	\mathcal{T}^{q^{a}}_m\,. \mathbb{T}^{q^{b}}_n={-\big(q^{a+b}-1\big)q^{-(m+1)b}\over \big(q^{a}-1\big)\big(q^{b}-1\big)}\mathbb{T}^{q^{a+b}}_{m+n+1}+{\mathbb{T}^{q^{a}}_{m+n+1}\over q^{b}-1}+ {q^{-(m+1)b}\over q^{a}-1}\mathbb{T}^{q^{b}}_{m+n+1} .
	\end{eqnarray}
\end{small}
Moreover, the following
commutation relations holds:
\begin{small}
	\begin{eqnarray}\label{scrtoAC}
	\Big[\mathcal{T}^{q^{a}}_m, \mathcal{T}^{q^{b}}_n\Big]&=&{\big(q^{a+b}-1\big)\big(q^{-na}-q^{-mb}\big)\over \big(q^{a}-1\big)\big(q^{b}-1\big)}\mathcal {T}^{q^{a+b}}_{m+n}-{\big(q^{-na}-1\big)\over q^{b}-1}\mathcal {T}^{q^{a}}_{m+n}\nonumber\\&+& {\big(q^{-mb}-1\big)\over q^{a}-1}\mathcal {T}^{q^{b}}_{m+n},
	\end{eqnarray}
\end{small}
\begin{small}
	\begin{eqnarray}\label{scrgoAC}
	\Big[{\mathcal T}^{q^{a}}_m, \mathbb{T}^{q^{b}}_n\Big]&=&{\big(q^{a+b}-1\big)\big(q^{-n\,a}-q^{-m\,b+a}\big)\over \big(q^{a}-1\big)\big(q^{b}-1\big)}\mathbb{T}^{q^{a+b}}_{m+n}\nonumber\\ &+&{\big(q^{-m\,b}\,q^{a}-1\big)\over q^{b}-1}\mathbb{T}^{q^{a}}_{m+n} + {\big(q^{-m\,b}-1\big)\over q^{a}-1}\mathbb {T}^{q^{b}}_{m+n} + f(m,n),
	\end{eqnarray}
\end{small}
where
\begin{eqnarray*}
	f(m,n)&=&-{\big(q^{a+b}-1\big)q^{-m\,b-b}\over \big(q^{a}-1\big)\big(q^{b}-1\big)}\mathbb{T}^{q^{a+b}}_{1}+{q^{1\over q^{b}-1}\mathbb{T}^{q^{a}}_{1}}+ {q^{-m\,b-b}\over q^{a}-1}\mathbb{T}^{q^{b}}_{1}
\end{eqnarray*}
and other anti-commutators are zeros.
Setting $a=b=1,$ we obtain:
\begin{small}
	\begin{eqnarray*}
		\Big[\mathcal{T}^{q}_m, \mathcal{T}^{q}_n\Big]={\big(q^{-n}-q^{-m}\big)\over \big(q-1\big)}\,[2]_{q}\mathcal{T}^{q^{2}}_{m+n}-{1\over q-1}\Big(\big(q^{-n}-1\big)-\big(q^{-m}-1\big)\Big) \mathcal{T}^{q}_{m+n},
	\end{eqnarray*}
\end{small}
\begin{small}
	\begin{eqnarray*}
		\Big[\mathcal{T}^{q}_m, \mathbb{T}^{q}_n\Big]={(q^{-n}-q^{-m+1})\over q-1}[2]_{q}\mathbb{T}^{q^{2}}_{m+n}+{1\over q-1}\big((q-q^{-n})-(q^{-m}-1)\big)\mathbb{T}^{q}_{m+n} + f(m,n),
	\end{eqnarray*}
\end{small}
where
\begin{eqnarray*}
	f(m,n)=-{q^{-m-1}\over \big(q-1\big)}\,[2]_{q}\mathbb{T}^{q^{2}}_{1} +{\big(1+q^{-m-1}\big)\over q-1}\mathbb{T}^{q}_{1}
\end{eqnarray*}
and other anti-commutators are zeros.
We study the case with the same $q^{a}.$ Then, 
putting $a=b$ in the relation (\ref{scrtoAC}), we obtain:
\begin{small}
	\begin{eqnarray*}
		\Big[\mathcal {T}^{q^{a}}_m, \mathcal{T}^{q^{a}}_n\Big]={\big(q^{-n\,a}-q^{-m\,a}\big)\over \big(q^{a}-1\big)}\,[2]_{q^{a}}\mathcal {T}^{q^{2\,a}}_{m+n}-{1\over q^{a}-1}\Big(\big(q^{-n\,a}+1\big)-\big(q^{-m\,a}+1\big)\Big) \mathcal{T}^{q^{a}}_{m+n}
	\end{eqnarray*}
\end{small}
and
\begin{small}
	\begin{eqnarray*}
		\Big[\mathcal{T}^{q^{a}}_{m_1},\cdots, \mathcal {T}^{q^{a}}_{m_n}\Big]&=&{(-1)^{n+1}\over (q^{a}-1)^{n-1}}\Big( M^n_{a}[n]_{q^{a}}\mathcal{T}^{q^{na}}_{m_1+\cdots+m_n}\nonumber\\&-& {[n-1]_{q^{a}}\over  q^{-a(\sum_{l=1}^{n}m_l)}}(M^n_{a}+ C^n_{a})\mathcal {T}^{q^{(n-1)a}}_{m_1+\cdots+m_n}\Big),
	\end{eqnarray*}
\end{small}
where 
\begin{small}
	\begin{eqnarray*}
		M^n_{a}&=& q^{-a(n-1)\sum_{s=1}^{n}m_s}\Big((q^{a}-1)^{n\choose 2}\prod_{1\leq j < k \leq n}\Big([m_k]_{q^{a}}-[m_j]_{q^{a}}\Big)\nonumber\\&+&\prod_{1\leq j < k \leq n}(q^{am_k}-q^{am_j})\Big)
	\end{eqnarray*}
	and 
	\begin{eqnarray*}
		C^{n}_{a}
		&=&q^{-a(n-1)\sum_{s=1}^{n}m_s}\Big((q^{a}-1)^{n\choose 2}\prod_{1\leq j < k \leq n}([m_k]_{q^{a}}-[m_j]_{q^{a}})\nonumber\\&+&(-1)^{n-1}\prod_{1\leq j < k \leq n}(q^{am_k}-q^{am_j})\Big).
	\end{eqnarray*}
\end{small}
From the super multibracket of order $n$ (\ref{smb}), we define the $q-$ deformed $n-$ bracket as follows:
\begin{small}
	\begin{eqnarray*}
		\big[\mathcal{T}^{q^{a}}_{m_1},\mathcal{T}^{q^{a}}_{m_2},\ldots, \mathbb{T}^{q^{a}}_{m_n}\big]:=\sum_{j=0}^{n-1}(-1)^{n-1+j}\epsilon^{i_1\ldots i_{n-1}}_{12\ldots n-1}\mathcal{T}^{q^{a}}_{m_{i_1}}\ldots \mathcal{T}^{q^{a}}_{m_{i_j}}\mathbb{T}^{q^{a}}_{m_{n}}\mathcal{T}^{q^{a}}_{m_{i_{j+1}}}\ldots \mathcal{T}^{q^{a}}_{m_{i_{n-1}}}.
	\end{eqnarray*}
\end{small}
From the relation (\ref{scrgoAC}) with $a=b,$ we obtain:
\begin{small}
	\begin{eqnarray}
	\Big[\mathcal {T}^{q^{a}}_m, \mathbb{T}^{q^{a}}_n\Big]&=&{(q^{-n\,a}-q^{-(m-1)a})\over (q^{a}-1)}[2]_{q^{a}}\mathbb{T}^{q^{2a}}_{m+n}+{1\over q^{a}-1}\big((q^{a}-q^{-na})\nonumber\\ &+&(q^{-ma}-1)\big)\mathbb{T}^{q^{a}}_{m+n}+ f(m,n),
	\end{eqnarray}
	where
	\begin{eqnarray*}
		f(m,n)&=&-{q^{-m\,a-a}\over \big(q^{a}-1\big)}[2]_{q^{a}}\mathbb{T}^{q^{2a}}_{1} + {1\over q^{a}-1}\bigg(1+q^{-m\,a-a}\bigg)\mathbb{T}^{q^{a}}_{1}.
	\end{eqnarray*}
\end{small}
Thus, the super $n-$ bracket takes the form:
\begin{small}
	\begin{eqnarray*}
		\Big[\mathcal{T}^{q^{a}}_{m_1},\cdots, \mathbb{T}^{q^{a}}_{m_n}\Big]&=&{(-1)^{n+1}\over \big(q^{a}-1\big)^{n-1}}\Big( A^n_{a}[n]_{q^{a}}\mathbb{T}^{q^{n\,a}}_{m_1+\cdots+m_n}\nonumber\\ &-& [n-1]_{q^{a}}\big(F^n_{a}+ S^n_{a}\big){\mathcal T}^{q^{(n-1)a}}_{m_1+\cdots+m_n}\Big)+ f\big(m_{1},\ldots, m_{n}\big),
	\end{eqnarray*}
\end{small}
where 
\begin{small}
	\begin{eqnarray*}
		A^n_{a}&=& q^{-a(n-1)\sum_{s=1}^{n}(m_s-1)}\big(q^{a}-1\big)^{n\choose 2}\prod_{1\leq j < k \leq n}\Big([m_k-1]_{q^{a}}-[m_j]_{q^{a}}\Big),
	\end{eqnarray*}
	\begin{eqnarray*}
		F^n_{a}&=& q^{-a(n-1)\sum_{s=1}^{n}m_s}\big(q^{a}-1\big)^{n\choose 2}\prod_{1\leq j < k \leq n}\Big([m_k]_{q^{a}}-[m_j]_{q^{a}}\Big),
	\end{eqnarray*} 
	\begin{small}
		\begin{eqnarray*}
			S^{n}_{a}
			&=&\big(q^{a}-1\big)^{n\choose 2}\prod_{1\leq j < k \leq n}\Big([m_k]_{q^{a}}-[m_j]_{q^{a}}p^{n\choose 2}\Big)+(-1)^{n-1}\prod_{1\leq j < k \leq n}\Big(q^{am_k}-q^{am_j}q^{n\choose 2}\Big)
		\end{eqnarray*}
	\end{small}
	and 
	\begin{eqnarray*}
		f\big(m_{1},\ldots, m_{n}\big)={(-1)^{n+1}q^{-(m+1)a}\over \big(q^{a}-1\big)^{n-1}}\Big( [n]_{q^{a}}\mathbb{T}^{q^{n\,a}}_{1}- {[2(m+1)]_{q^{a}}\over [m+1]_{q^{a}}}\mathbb{T}^{q^{(n-1)a}}_{1}\Big).
	\end{eqnarray*}
\end{small}
The operators (\ref{stoAC}) and (\ref{sgoAC}) take the following forms:
\begin{small}
	\begin{eqnarray*}
		\mathcal{T}^{q^{a}}_{m}&=&[m+\gamma]_{q^{a}}m!{\partial\over \partial t_m}+ {q^{m+\gamma}\over q^{a} - q^{-a}}\sum_{k=1}^{\infty}{(k+m)!\over k!}B_k(t^{a}_1,\cdots,t^{a}_k){\partial\over \partial t_{k+m}}\\
		\,\mathbb{T}^{q^{a}}_{m}&=&\theta\bigg([m+\gamma]_{q^{a}}m!{\partial\over \partial t_m}+ {q^{m+\gamma}\over q^{a} - q^{-a}}\sum_{k=1}^{\infty}{(k+m)!\over k!}B_k(t^{a}_1,\cdots,t^{a}_k){\partial\over \partial t_{k+m}}\bigg).
	\end{eqnarray*}
\end{small} 
Putting $\bar{m}=m+\gamma,\quad \bar{n}=n+\gamma,$ and by changing $n!\,{\partial\over \partial t_n}\longleftrightarrow x^n,$ we show directly that  the products $\mathcal{T}^{q^{a}}_{m}\,.\mathcal{T}^{q^{b}}_{n}$ and $\mathcal{T}^{q^{a}}_{m}\,.\mathbb{T}^{q^{b}}_{n}$ are respectively equivalent to (\ref{ACprod1}) and (\ref{ACprod2}).
\section{Relevant particular cases}
	Particular cases of super Virasoro $n-$ algebra and application associated to different quantum algebras in the literature are deduced as follows:  
	\subsection{{\bf Jagannathan- Srinivasa} deformation \cite{JS}}
			Taking ${\mathcal R}(x,y)={x-y\over p-q},$	we obtain: 
			the algebra endomorphism $\sigma$ on $\mathcal{B}$ is defined by:
			\begin{eqnarray*}
				\sigma(t^n):=\big(p\,q\big)^n\,t^n\quad\mbox{and}\quad \sigma(\theta):=(p\,q)\,\theta.
			\end{eqnarray*}
			We define also the two $(p,q)-$ deformed linear maps by:
			\begin{eqnarray*}
				\left \{
				\begin{array}{l}
					\partial_t(t^n)=[n]_{p,q}\,t^n\mbox{,}\quad \partial_t(\theta\,t^n)=[n]_{p,q}\,\theta\,t^n, \\
					\\
					\partial_{\theta}(t^n)=0\mbox{,}\quad \partial_{\theta}(\theta\,t^n)=\big(p\,q\big)^n\,t^n.
				\end{array}
				\right .
			\end{eqnarray*}
			The linear map $\Delta=\partial_{t}+\theta \partial_{\theta}$ on  ${\mathcal B}$ is an even $\sigma$-derivation, and satisfy the following relations:
			\begin{eqnarray*}
				\,\Delta(x\,y)&=&\Delta(x)\,y+\sigma(x)\Delta(y),\nonumber\\
				\,\Delta(t^n)&=& [n]_{p,q}\,t^n\quad\mbox{and}\quad \Delta(\theta\,t^n)= \big([n]_{p,q} + \big(p\,q\big)^n\big)\,\theta\,t^n. 
			\end{eqnarray*}		
			It  is generated by bosonic and fermionic operators $l^{p,q}_m=-t^m\,\Delta$ of parity $0$ and $G^{p,q}_m=-\theta\,t^m\,\Delta$ of parity $1$ verifying the following commutations  relations:
			\begin{eqnarray*}
				\,\big[l^{p,q}_{m_1},l^{p,q}_{m_2}\big]_{\hat{x}, \hat{y}}&=&\big([m_1]_{p,q}-[m_2]_{p,q}\big)\,l^{p,q}_{m_1+m_2},\nonumber\\
				\,\big[l^{p,q}_{m_1}, G^{p,q}_{m_2}\big]_{x,y}&=&\big([m_1]_{p,q}-[m_2+1]_{p,q}\big)\,G^{p,q}_{m_1+m_2},\nonumber\\
				\,\big[G^{p,q}_{m_1},G^{p,q}_{m_2}\big]&=&0,
			\end{eqnarray*}
			where\begin{eqnarray}\label{JScoefcom}
			\left \{
			\begin{array}{l}
			\hat{x}=\chi_{m_1m_2}(p,q)\mbox{,}\quad \hat{y}=(pq)^{m_2-m_1}\,\chi_{m_1m_2}(p,q), \\
			\\
			x=\tau_{m_1m_2}\mbox{,}\quad y=(pq)^{1+m_2-m_1}\,\tau_{m_1m_2}, \\
			\\
			\chi_{m_1m_2}(p,q)={[m_1]_{p,q}-[m_2]_{p,q}\over (pq)^{m_2-m_1}\,[m_1]_{p,q}-[m_2]_{p,q}}\\
			\\
			\tau_{m_1m_2}(p,q)={[m_1]_{p,q}-[m_2+1]_{p,q}\over (pq)^{1+m_2-m_1}\,[m_1]_{p,q}-[m_2]_{p,q}-(pq)^{m_2}}.
			\end{array}
			\right .
			\end{eqnarray}
			The $(p,q)-$ deformed $n-$ bracket $(n\geq 3)$ are defined  as follows:
			\begin{small}
				\begin{eqnarray*}
					\big[l^{p,q}_{m_1},\ldots, l^{p,q}_{m_n}\big]&:=&\bigg({p^{-\sum_{l=1}^{n}m_l}+q^{-\sum_{l=1}^{n}m_l}\over 2}\bigg)^{\alpha}\epsilon^{i_1i_2\cdots i_n}_{12\cdots n}\nonumber\\&\times&(p\,q)^{\sum_{j=1}^{n}\big(\lfloor {n\over 2} \rfloor-j+1\big)m_{i_j}}l^{p,q}_{m_{i_1}}\ldots
					l^{p,q}_{m_{i_n}},
				\end{eqnarray*}
			\end{small}
			and
			\begin{small}
				\begin{eqnarray*}
					\big[l^{p,q}_{m_1},l^{p,q}_{m_2},\ldots, G^{p,q}_{m_n}\big]&:=&\bigg({p^{-\sum_{l=1}^{n}m_l}+q^{-\sum_{l=1}^{n}m_l}\over 2}\bigg)^{\alpha}\sum_{j=0}^{n-1}(-1)^{n-1+j}\epsilon^{i_1\ldots i_{n-1}}_{12\ldots n-1}\nonumber\\
					&\times&(p\,q)^{\beta}l^{p,q}_{m_{i_1}}\ldots l^{p,q}_{m_{i_j}}
					G^{p,q}_{m_{n}}l^{p,q}_{m_{i_{j+1}}}\ldots l^{p,q}_{m_{i_{n-1}}},
				\end{eqnarray*}
			\end{small}
			where $\beta=\sum_{k=1}^{j}\big(\lfloor {n\over 2} \rfloor-k+1\big)m_{i_k}+\big(\lfloor  {n\over 2} \rfloor-1\big)\big(m_n+1\big)+\sum_{k=j+1}^{n-1}\big(\lfloor {n\over 2} \rfloor -k\big)m_{i_k},$  $\alpha={1+(-1)^n \over 2},$    and $\lfloor n \rfloor=Max\{m\in\mathbb{Z}\ m\leq n\}$ is the floor function. Then, 
			the  generators $l^{p,q}_{m}$ and $G^{p,q}_{m}$ satisfy the commutation relations:
			\begin{eqnarray*}
				\big[l^{p,q}_{m_1},l^{p,q}_{m_2},\ldots, l^{p,q}_{m_n}\big]&=&{\big(q-p\big)^{n-1\choose 2}\over (p\,q)^{\lfloor {n-1\over 2}\rfloor\sum_{l=1}^{n}m_l}}\Big({p^{-\sum_{l=1}^{n}m_l}+q^{-\sum_{l=1}^{n}m_l}\over 2}\Big)^{\alpha}\nonumber\\&\times&\prod_{1\leq i < j\leq n}\Big([m_i]_{p,q}-[m_j]_{p,q}\Big)l_{\sum_{l=1}^{n}m_l},
			\end{eqnarray*}
			\begin{small}
				\begin{eqnarray*}
					\big[l^{p,q}_{m_1},l^{p,q}_{m_2},\cdots,G^{p,q}_{m_n}\big]&=& {\big(q-p\big)^{{n-1\choose 2}}\over \big(pq\big)^{\lfloor {n-1\over 2}\rfloor \sum_{l=1}^{n}m_l+1}}\Big({p^{-\sum_{l=1}^{n}m_l-1}+q^{-\sum_{l=1}^{n}m_l-1}\over 2}\Big)^{\alpha}\nonumber\\&\times&
					\prod_{1\leq i<j\leq n-1} \big([m_i]_{p,q}- [m_j]_{p,q}\big)\prod_{i=1}^{n-1} \big([m_i]_{p,q}- [m_n+1]_{p,q}\big)G^{p,q}_{\sum_{l=1}^{n}m_l}
				\end{eqnarray*}
			\end{small}
			and other anti-commutators are zeros. Furthermore, the corresponding
			   Virasoro $2n-$ algebra  is deduced as: \begin{eqnarray*}
				\big[L_{m_1},\cdots,L_{m_{2n}}\big] =g_{p,q}(m_1,\cdots,m_{2n}) + C_{p,q}(m_1,\cdots,m_{2n}),
			\end{eqnarray*}
			where
			\begin{small}
				\begin{eqnarray}\label{jsgv}
				g_{p,q}(m_1,\cdots,m_{2n})&=&{(q-p)^{{2n-1\choose 2}}\over 2(pq)^{(n-1) \sum_{l=1}^{2n}m_l}}\Big(p^{-\sum_{l=1}^{2n}m_l}+ q^{-\sum_{l=1}^{2n}m_l}\Big)\nonumber\\
				&\times&\prod_{1\leq i< j\leq 2n}\Big([m_i]_{p,q}-[m_j]_{p,q}\Big)L_{\sum_{l=1}^{2n}m_l}
				\end{eqnarray}
				and
				\begin{eqnarray}\label{jscv}
				C_{p,q}(m_1,\cdots,m_{2n})&=&{c(p,q)\epsilon^{i_1\cdots i_{2n}}_{1\cdots 2n}\over 6\times  2^n\times n!}\prod_{l=1}^{n}{[m_{i_{2l-1}}-1]_{p,q}\over (pq)^{m_{2l-1}}\big(p^{m_{i_{2l-1}}}+q^{m_{i_{2l-1}}}\big)}\nonumber\\&\times& [m_{i_{2l-1}}]_{p,q}[m_{i_{2l-1}}+1]_{p,q}
				\delta_{m_{i_{2l-1}}+ m_{i_{2l}},0.}
				\end{eqnarray}	
			\end{small}
				Several examples are deduced as follows:
			\begin{enumerate}
				\item [(a)]
				Taking $n=2$ in the realtions (\ref{jsgv}) and (\ref{jscv}), we obtain the $(p,q)-$ deformed Virasoro $4-$ algebra:
				\begin{small}
					\begin{eqnarray*}
						\big[L_{m_1},L_{m_2},L_{m_3},L_{m_{4}}\big] =g_{p,q}(m_1,m_2,m_3,m_{4})+C_{p,q}(m_1,\cdots,m_{4}),
					\end{eqnarray*}
					where
					\begin{eqnarray*}
						g_{p,q}(m_1,m_2,m_3,m_4) &=& {(q-p)^{3}\over \big(pq\big)^{ m_1+m_2+m_3+m_4}}\Big(p^{-\sum_{l=1}^{4}m_l}+ q^{-\sum_{l=1}^{4}m_l}\Big)\nonumber\\&\times&\prod_{1\leq i < j\leq 4}\Big([m_i]_{p,q}-[m_j]_{p,q}\Big)L_{\sum_{l=1}^{4}m_l}
					\end{eqnarray*}
					and
					\begin{eqnarray*}
						C_{p,q}(m_1,\cdots,m_{4})&=&{c(p,q)\,\epsilon^{i_1\cdots i_{4}}_{1\cdots 4}\over 48}\prod_{l=1}^{2}\big(p\,q\big)^{-m_{2l-1}}{[m_{2l-1}]_{p,q}\over [2\,m_{2l-1}]_{p,q}}\nonumber\\&\times& [m_{i_{2l-1}}-1]_{p,q}[m_{i_{2l-1}}]_{p,q}[m_{i_{2l-1}}+1]_{p,q}
						\,\delta_{m_{i_{2l-1}}+ m_{i_{2l}},0.}
					\end{eqnarray*}
				\end{small}
				\item [(b)]The $(p,q)-$ deformed Virasoro $6-$ algebra is deduced from the generalization by taking $n=3:$
				\begin{eqnarray*}
					\big[L_{m_1},\cdots,L_{m_{6}}\big] =g_{p,q}(m_1,\cdots,m_{6})+C_{p,q}(m_1,\cdots,m_{6}),
				\end{eqnarray*}
				where
				\begin{small}
					\begin{eqnarray*}
						g_{p,q}(m_1,\cdots,m_{6})&=&{(q-p)^{10}\over (pq)^{2 \sum_{l=1}^{6}m_l}}\Big(p^{-\sum_{l=1}^{6}m_l}+ q^{-\sum_{l=1}^{6}m_l}\Big)\nonumber\\&\times&\prod_{1\leq i< j\leq 6}\Big([m_i]_{p,q}-[m_j]_{p,q}\Big)L_{\sum_{l=1}^{6}m_l}
					\end{eqnarray*}
					and
					\begin{eqnarray*}
						C_{p,q}(m_1,\cdots,m_{6})&=&{c(p,q)\epsilon^{i_1\cdots i_6}_{1\cdots 6}\over 288}\prod_{l=1}^{3}\big(pq\big)^{-m_{2l-1}}{[m_{2l-1}]_{p,q}\over [2m_{2l-1}]_{p,q}}\nonumber\\&\times& [m_{i_{2l-1}}-1]_{p,q}[m_{i_{2l-1}}]_{p,q}[m_{i_{2l-1}}+1]_{p,q}
						\,\delta_{m_{i_{2l-1}}+ m_{i_{2l}},0.}
					\end{eqnarray*}
				\end{small}
			\end{enumerate}
				The $(p,q)-$ deformed  super Jacobi identity  is given by :\begin{eqnarray*}
						\sum_{(i,j,l)\in\mathcal{C}(n,m,k)}\,(-1)^{|A_i||A_l|}\big[\rho(A_i),\big[A_j,A_l\big]_{p,q}\big]_{p,q}=0,
					\end{eqnarray*} 
					where  $\rho(l^{p,q}_{m_1})=(p^{m_1}+q^{m_1})\,l^{p,q}_{m_1},$ $\rho(G^{p,q}_{m_1})=(p^{m_1+1}+q^{m_1+1})\,G^{p,q}_{m_1} $ and $\mathcal{C}(n,m,k)$ denotes the cyclic permutation of $(n,m,k)$.
					
					Moreover, 
 the operators $\bar{l}^{p,q}_{m}$ and $\bar{G}^{p,q}_{m}$ satisfy  the following commutation relations:
					\begin{small}
						\begin{eqnarray*}
							\big[\bar{l}^{p,q}_{m_1},\bar{l}^{p,q}_{m_2}\big]_{\hat{x}, \hat{y}}&=&\big([m_1]_{p,q}-[m_2]_{p,q}\big)\,\bar{l}^{p,q}_{m_1+m_2} + {c(p,q)(pq)^{m_1}[m_1]_{p,q}\over 6[2m_1]_{p,q}}\nonumber\\&\times& [m_1+1]_{p,q}[m_1]_{p,q}[m_1-1]_{p,q}\delta_{m_1+m_2,0}, \end{eqnarray*}
						and
						\begin{eqnarray*}
							\big[\bar{l}^{p,q}_{m_1}, \bar{G}^{p,q}_{m_2}\big]_{x,y}&=&\big([m_1]_{p,q}-[m_2+1]_{p,q}\big)\,\bar{G}^{p,q}_{m_1+m_2} + {c(p,q)(pq)^{m_1}\,[m_1]_{p,q}\over 6[2m_1]_{p,q}}\nonumber\\&\times& [m_1+1]_{p,q}[m_1]_{p,q}[m_1-1]_{p,q}\,\delta_{m_1+m_2+1,0},
						\end{eqnarray*}
					\end{small}
					where $\hat{x},$ $\hat{y},$ $x,$ and $y$ are given by the relation (\ref{JScoefcom})
			The super Virasoro $2n-$ algebra  is presented as follows: \begin{small}
				\begin{eqnarray*}
					\big[\bar{L}^{p,q}_{m_1},\cdots,\bar{L}^{p,q}_{m_{2n}}\big] =g_{p,q}(m_1,\cdots,m_{2n}) + C_{p,q}(m_1,\cdots,m_{2n}),
				\end{eqnarray*}
				\begin{eqnarray*}
					\big[\bar{L}^{p,q}_{m_1},\bar{L}^{p,q}_{m_2},\cdots, \bar{G}^{p,q}_{m_{2n}}\big]= f_{p,q}(m_1,m_2,\cdots m_{2n})+ {\mathcal CS}_{p,q}(m_1,\cdots m_{2n}),
				\end{eqnarray*}
			\end{small}
			where $g_{p,q}(m_1,\cdots,m_{2n})$ and $C_{p,q}(m_1,\cdots,m_{2n})$ are given by the relations (\ref{jsgv}), (\ref{jscv}), 
			\begin{small}
				\begin{eqnarray*}
					f_{p,q}(m_1,\cdots m_{2n})&=&{\big(q-p\big)^{{2n-1\choose 2}}\over 2(pq)^{-(n-1) \sum_{l=1}^{2n}m_l+1}}\Big( p^{\sum_{l=1}^{2n}m_l-1}+q^{\sum_{l=1}^{2n}m_l-1}\Big)\nonumber\\&\times&\prod_{1\leq i<j\leq 2n-1} \big([m_i]_{p,q}- [m_j]_{p,q}\big)
					\prod_{i=1}^{2n-1} \big([m_i]_{p,q}- [m_{2n}+1]_{p,q}\big)G_{\sum_{l=1}^{2n}m_l},
				\end{eqnarray*}
			\end{small}
			\begin{small}
				\begin{eqnarray*}
					{\mathcal CS}_{p,q}(m_1,m_2,\cdots m_{2n})&=&\sum_{k=1}^{2n-1}{(-1)^{k+1}c(p,q)(pq)^{-m_k}\over 6\times 2^{n-1}(n-1)!}{1\over p^{m_k}+q^{m_k}}\nonumber\\&\times&[m_k+1]_{p,q}[m_k]_{p,q}[m_k-1]_{p,q}\delta_{m_k+m_{2n}+1,0}\nonumber\\&\times&\epsilon^{i_1\cdots i_{2n-2}}_{j_1\cdots j_{2n-2}}\prod_{s=1}^{n-1}{(pq)^{-i_{2s-1}}\over p^{i_{2s-1}}+q^{i_{2s-1}}}\nonumber\\&\times&[i_{2s-1}+1]_{p,q}\,[i_{2s-1}]_{p,q}\,[i_{2s-1}-1]_{p,q}\delta_{i_{2s-1}+i_{2s},0},
				\end{eqnarray*}
			\end{small}
			with $\{j_1,\cdots, j_{2n-2}\}=\{1,\cdots,\hat{k},\cdots,2n-1\}$ and other anti-commutators are zeros.
			
			 Now, we construct another $(p,q)-$ deformed super Witt $n-$ algebra.    
			We consider  the operators defined by:
			\begin{eqnarray}\label{stopq}
			\,{\mathcal T}^{p^{a},q^{a}}_m&=&\Delta\,z^{m},\\
			\,\mathbb{T}^{p^{a},q^{a}}_m&=&-\theta\,\Delta\,z^{m}\label{sgopq}.
			\end{eqnarray}
			
			The operators (\ref{stopq}) and (\ref{sgopq}) can be rewritten as:
			\begin{eqnarray*}
				\,\mathcal{T}^{p^{a},q^{a}}_m&=&-[m]_{p^{a},q^{a}}\,z^{m}\\
				\,\mathbb{T}^{p^{a},q^{a}}_m&=&-\theta\,[m]_{p^{a},q^{a}}\,z^{m}.
			\end{eqnarray*}
			The following products hold.
			\begin{eqnarray}\label{JSprod1}
			\mathcal{T}^{p^{a},q^{a}}_m\,. \mathcal{T}^{p^{b},q^{b}}_n&=&-{\big(p^{a+b}-q^{a+b}\big)p^{-mb}\over \big(p^{a}-q^{a}\big)\big(p^{b}-q^{b}\big)}{\mathcal T}^{p^{a+b},q^{a+b}}_{m+n}\nonumber\\ &+&{q^{-nb}\over p^{b}-q^{b}}{\mathcal T}^{p^{a},q^{a}}_{m+n} + {q^{(m+n)a}p^{-mb}\over p^{a}-q^{a}}{\mathcal T}^{p^{b},q^{b}}_{m+n}
			\end{eqnarray}
			and
			\begin{eqnarray}\label{JSprod2}
			\mathcal{T}^{p^{a},q^{a}}_m\,. \mathbb{T}^{p^{b},q^{b}}_n&=&-{\big(p^{a+b}-q^{a+b}\big)p^{-(m+1)b}\over \big(p^{a}-q^{a}\big)\big(p^{b}-q^{b}\big)}\mathbb{T}^{p^{a+b},q^{a+b}}_{m+n+1}\nonumber\\ &+&{q^{-n\,b}\over p^{b}-q^{b}}\mathbb{T}^{p^{a},q^{a}}_{m+n+1}+ {q^{(m+n+1)\,a}p^{-(m+1)b}\over p^{a}-q^{a}}\mathbb{T}^{p^{b},q^{b}}_{m+n+1} .
			\end{eqnarray}
			and the operators satisfy the following
			commutation relations
			\begin{small}
				\begin{eqnarray}\label{scrtopq}
				\Big[\mathcal{T}^{p^{a},q^{a}}_m, \mathcal{T}^{p^{b},q^{b}}_n\Big]&=&\frac{\big(p^{a+b}-q^{a+b}\big)\big(p^{-na}-p^{-mb}\big)}{ \big(p^{a}-q^{a}\big)\big(p^{b}-q^{b}\big)}\mathcal {T}^{p^{a+b},q^{a+b}}_{m+n}\nonumber\\ &-&{q^{(m+n)b}\big(p^{-na}-q^{-mb}\big)\over p^{b}-q^{b}}\mathcal {T}^{p^{a},q^{a}}_{m+n}+ {q^{(m+n)a}\big(p^{-mb}-q^{-na}\big)\over p^{\alpha}-q^{a}}\mathcal {T}^{p^{b},q^{b}}_{m+n},
				\end{eqnarray}
			\end{small}
			\begin{small}
				\begin{eqnarray}\label{scrgopq}
				\Big[{\mathcal T}^{p^{a},q^{a}}_m, \mathbb{T}^{p^{b},q^{b}}_n\Big]&=&{\big(p^{a+b}-q^{a+b}\big)\big(p^{-n\,a}-p^{-m\,b+a}\big)\over \big(p^{a}-q^{a}\big)\big(p^{b}-q^{b}\big)}\mathbb{T}^{p^{a+b},q^{a+b}}_{m+n}\nonumber\\ &+&{q^{(m+n)b}\big(q^{-m\,b}\,p^{a}-p^{-n\,a}\big)\over p^{b}-q^{b}}\mathbb{T}^{p^{a},q^{a}}_{m+n}\nonumber\\ &+& {q^{(m+n)\,a}\big(p^{-m\,b}q^{a}-q^{-n\,a}\big)\over p^{a}-q^{a}}\mathbb {T}^{p^{b},q^{b}}_{m+n} + f(m,n),
				\end{eqnarray}
			\end{small}
			where
			\begin{eqnarray*}
				f(m,n)&=&-{\big(p^{a+b}-q^{a+b}\big)p^{-m\,b-b}\,q^{(a+b)(m+n)}\over \big(p^{a}-q^{a}\big)\big(p^{b}-q^{b}\big)}\mathbb{T}^{p^{a+b},q^{a+b}}_{1}\nonumber\\ &+&{q^{(m+n)a}\,q^{n\,b}\over p^{b}-q^{b}}\mathbb{T}^{p^{a},q^{a}}_{1}+ {q^{(m+n)(a+b)}p^{-m\,b-b}\,q^{a}\over p^{a}-q^{a}}\mathbb{T}^{p^{b},q^{b}}_{1}
			\end{eqnarray*}
			and other anti-commutators are zeros.
			
			Setting $a=b=1,$ we obtain:
			\begin{small}
				\begin{eqnarray*}
					\Big[\mathcal{T}^{p,q}_m, \mathcal{T}^{p,q}_n\Big]=\frac{\big(p^{-n}-p^{-m}\big)}{\big(p-q\big)}[2]_{p,q}\mathcal{T}^{p^{2},q^{2}}_{m+n}-\frac{q^{m+n}}{p-q}\Big(\big(p^{-n}-q^{-m}\big)-\big(p^{-m}-q^{-n}\big)\Big) \mathcal{T}^{p,q}_{m+n},
				\end{eqnarray*}
				\begin{eqnarray*}
					\Big[\mathcal{T}^{p,q}_m, \mathbb{T}^{p,q}_n\Big]&=&{(p^{-n}-p^{-m+1})\over p-q}\,[2]_{p,q}\mathbb{T}^{p^{2},q^{2}}_{m+n}+f(m,n)\nonumber\\ &+&{q^{m+n}\over p-q}\bigg((q^{-m}\,p-p^{-n})-(p^{-m}\,q-q^{-n})\bigg)\mathbb{T}^{p,q}_{m+n} ,
				\end{eqnarray*}
			\end{small}
			where
			\begin{eqnarray*}
				f(m,n)=-{p^{-m-1}\,q^{2(m+n)}\over \big(p-q\big)}\,[2]_{p,q}\mathbb{T}^{p^{2},q^{2}}_{1} +{q^{m+n}\big(q^{n}+q^{m+n}\,p^{-m-1}\,q\big)\over p-q}\mathbb{T}^{p,q}_{1}
			\end{eqnarray*}
			and other anti-commutators are zeros.
			
			We consider the $n-$ bracket defined by:
			\begin{eqnarray*}
				\Big[\mathcal {T}^{p^{a_1},q^{a_1}}_{m_1},\cdots,\mathcal{T}^{p^{a_n},q^{a_n}}_{m_n}
				\Big]:=\epsilon^{i_1 \cdots i_n}_{1 \cdots n}\,\mathcal{T}^{p^{a_{i_1}},q^{a_{i_1}}}_{m_{i_1}} \cdots \mathcal {T}^{p^{a_{i_n}},q^{a_{i_n}}}_{m_{i_n}}.
			\end{eqnarray*}
				We study the case with the same $(p^{a},q^{a}).$ Then, 
				\begin{eqnarray*}
					\Big[\mathcal {T}^{p^{a},q^{a}}_{m_1},\cdots,\mathcal{T}^{p^{a},q^{a}}_{m_n}
					\Big]=\epsilon^{1\cdots n}_{1\cdots n}\,\mathcal {T}^{p^{a},q^{a}}_{m_{1}}\cdots \mathcal{T}^{p^{a},q^{a}}_{m_{n}}.
				\end{eqnarray*}
				Putting $a=b$ in the relation (\ref{scrtopq}), we obtain:
				\begin{small}
					\begin{eqnarray*}
						\Big[\mathcal {T}^{p^{a},q^{a}}_m, \mathcal{T}^{p^{a},q^{a}}_n\Big]&=&{\big(p^{-n\,a}-p^{-m\,a}\big)\over \big(p^{a}-q^{a}\big)}\,[2]_{p^{a},q^{a}}\mathcal {T}^{p^{2\,a},q^{2\,a}}_{m+n}\nonumber\\&-&{\tau^{(m+n)a}_2\over p^{a}-q^{a}}\Big(\big(p^{-n\,a}-p^{-m\,a}\big)+\big(q^{-n\,a}-q^{-m\,a}\big)\Big) \mathcal{T}^{p^{a},q^{a}}_{m+n}.
					\end{eqnarray*}
				\end{small}
				The $n-$ bracket takes the following form:
				\begin{small}
					\begin{eqnarray*}
						\Big[\mathcal{T}^{p^{a},q^{a}}_{m_1},\cdots, \mathcal {T}^{p^{\alpha},q^{\alpha}}_{m_n}\Big]&=&{(-1)^{n+1}\over \big(p^{a}-q^{a}\big)^{n-1}}\Big( M^n_{\alpha}[n]_{p^{a},q^{a}}\mathcal{T}^{p^{n\alpha},q^{n\alpha}}_{m_1+\cdots+m_n}\nonumber\\ &-& {[n-1]_{p^{\alpha},q^{\alpha}}\over  q^{-a\big(\sum_{l=1}^{n}m_l+1\big)}}\big(M^n_{a}+ C^n_{a}\big)\mathcal {T}^{p^{(n-1)a},q^{(n-1)a}}_{m_1+\cdots+m_n}\Big),
					\end{eqnarray*}
				\end{small}
				where 
				\begin{small}
					\begin{eqnarray*}
						M^n_{a}&=& p^{-a(n-1)\sum_{s=1}^{n}m_s}\Big(\big(p^{a}-q^{a}\big)^{n\choose 2}\prod_{1\leq j < k \leq n}\Big([m_k]_{p^{a},q^{a})}-[m_j]_{p^{a},q^{a}}\Big)\nonumber\\&+&\prod_{1\leq j < k \leq n}\Big(q^{a\,m_k}-q^{a\,m_j}\Big)\Big)
					\end{eqnarray*}
					and 
					\begin{eqnarray*}
						C^{n}_{a}
						&=&q^{-a(n-1)\sum_{s=1}^{n}m_s}\Big(\big(p^{a}-q^{a}\big)^{n\choose 2}\prod_{1\leq j < k \leq n}\Big([m_k]_{p^{a},q^{a}}-[m_j]_{p^{a},q^{a}}\Big)\nonumber\\&+&(-1)^{n-1}\prod_{1\leq j < k \leq n}\Big(p^{a\,m_k}-p^{a\,m_j}\Big)\Big).
					\end{eqnarray*}
				\end{small}
				From the super multibracket of order $n$ (\ref{smb}), we define the $(p,q)-$ deformed super $n-$ bracket as follows:
				\begin{small}
					\begin{eqnarray*}
						\big[\mathcal{T}^{p^{a},q^{a}}_{m_1},\mathcal{T}^{p^{a},q^{a}}_{m_2},\ldots, \mathbb{T}^{p^{a},q^{a}}_{m_n}\big]&:=&\sum_{j=0}^{n-1}(-1)^{n-1+j}\epsilon^{i_1\ldots i_{n-1}}_{12\ldots n-1}\mathcal{T}^{p^{a},q^{a}}_{m_{i_1}}\ldots \mathcal{T}^{p^{a},q^{a}}_{m_{i_j}}\nonumber\\&\times&
						\mathbb{T}^{p^{a},q^{a}}_{m_{n}}\mathcal{T}^{p^{a},q^{a}}_{m_{i_{j+1}}}\ldots \mathcal{T}^{p^{a},q^{a}}_{m_{i_{n-1}}}.
					\end{eqnarray*}
				\end{small}
				Using the relation (\ref{scrgopq}) with $a=b,$ we obtain:
				\begin{small}
					\begin{eqnarray*}
						\Big[\mathcal {T}^{p^{a},q^{a}}_m, \mathbb{T}^{p^{a},q^{a}}_n\Big]&=&{\big(p^{-n\,a}-p^{-(m-1)a}\big)\over \big(p^{a}-q^{a}\big)}[2]_{p^{a},q^{a}}\mathbb{T}^{p^{2a},q^{2a}}_{m+n}+ f(m,n)\nonumber\\ &+&{q^{(m+n)a}\over p^{a}-q^{a}}\bigg(\big(q^{-m\,a}\,p^{a}-p^{-n\,a}\big)+\big(p^{-m\,a}\,q^{a}-q^{-n\,a}\big)\bigg)\mathbb{T}^{p^{a},q^{a}}_{m+n},
					\end{eqnarray*}
					where
					\begin{eqnarray*}
						f(m,n)=-{p^{-ma-a}q^{2a(m+n)}\over \big(p^{a}-q^{a}\big)}[2]_{p^{a},q^{a}}\mathbb{T}^{p^{2a},q^{2a}}_{1}+{q^{(m+n)a}\over p^{a}-q^{a}}\bigg(q^{na}+{q^{(m+n+1)a}\over p^{ma+a}}\bigg)\mathbb{T}^{p^{a},q^{a}}_{1}.
					\end{eqnarray*}
				\end{small}
				Thus, the super $n-$ bracket takes the form:
				\begin{small}
					\begin{eqnarray*}
						\Big[\mathcal{T}^{p^{a},q^{a}}_{m_1},\cdots, \mathbb{T}^{p^{a},q^{a}}_{m_n}\Big]&=&{(-1)^{n+1}\over \big(p^{a}-q^{a}\big)^{n-1}}\Big( A^n_{a}[n]_{p^{a},q^{a}}\mathbb{T}^{p^{n\,a},q^{n\,a}}_{m_1+\cdots+m_n}\nonumber\\ &-& {[n-1]_{p^{a},q^{a}}\over  q^{-a\big(\sum_{l=1}^{n}m_l\big)}}\big(F^n_{a}+ S^n_{a}\big){\mathcal T}^{p^{(n-1)a},q^{(n-1)a}}_{m_1+\cdots+m_n}\Big)+ f\big(m_{1},\ldots, m_{n}\big),
					\end{eqnarray*}
				\end{small}
				where 
				\begin{small}
					\begin{eqnarray*}
						A^n_{a}&=& p^{-a(n-1)\sum_{s=1}^{n}(m_s-1)}\Big(\big(p^{a}-q^{a}\big)^{n\choose 2}\prod_{1\leq j < k \leq n}\Big([m_k-1]_{p^{a},q^{a}}-[m_j]_{p^{a},q^{a}}\Big)\nonumber\\&+&\prod_{1\leq j < k \leq n}\Big(q^{a(m_k-1)}-q^{a\,m_j}\Big)\Big),
					\end{eqnarray*}
					\begin{eqnarray*}
						F^n_{a}&=& p^{-a(n-1)\sum_{s=1}^{n}m_s}\Big(\big(p^{a}-q^{a}\big)^{n\choose 2}\prod_{1\leq j < k \leq n}\Big([m_k]_{p^{a},q^{a}}-[m_j]_{p^{a},q^{a}}q^{n\choose 2}\Big)\nonumber\\&+&\prod_{1\leq j < k \leq n}\Big(q^{a\,m_k}-q^{a\,m_j}q^{n\choose 2}\Big)\Big),
					\end{eqnarray*} 
					\begin{eqnarray*}
						S^{n}_{a}
						&=&q^{-a(n-1)\sum_{s=1}^{n}m_s}\Big(\big(p^{a}-q^{a}\big)^{n\choose 2}\prod_{1\leq j < k \leq n}\Big([m_k]_{p^{a},q^{a}}-[m_j]_{p^{a},q^{a}}p^{n\choose 2}\Big)\nonumber\\&+&(-1)^{n-1}\prod_{1\leq j < k \leq n}\Big(p^{a\,m_k}-p^{a\,m_j}p^{n\choose 2}\Big)\Big)
					\end{eqnarray*}
					and 
					\begin{eqnarray*}
						f\big(m_{1},\ldots, m_{n}\big)&=&{(-1)^{n+1}p^{-(m+1)a}q^{a\sum_{l=1}^{n}m_l}\over \big(p^{a}-q^{a}\big)^{n-1}}\Big( q^{am}[n]_{p^{a},q^{a}}\mathbb{T}^{p^{na},q^{na}}_{1}\nonumber\\&-& \frac{[2(m+1)]_{p^{a},q^{a}}}{ [m+1]_{p^{a},q^{a}}}\mathbb{T}^{p^{(n-1)a},q^{(n-1)a}}_{1}\Big).
					\end{eqnarray*}
				\end{small}
				
				Furthermore, the operators (\ref{stopq}) and (\ref{sgopq}) are presented as follows:
				\begin{small}
					\begin{eqnarray*}
						\mathcal{T}^{{\mathcal R}(p^{a},q^{a})}_{m}&=&[m+\gamma]_{p^{a},q^{a}}m!{\partial\over \partial t_m} + {(pq)^{m+\gamma}\over p^{a} - q^{a}}\sum_{k=1}^{\infty}{(k+m)!\over k!}B_k(t^{a}_1,\cdots,t^{a}_k){\partial\over \partial t_{k+m}}\nonumber\\
						\mathbb{T}^{p^{a},q^{a}}_m&=&\theta\bigg([m+\gamma]_{p^{a},q^{a}}m!{\partial\over \partial t_m} + {(pq)^{m+\gamma}\over p^{a} - q^{a}}\sum_{k=1}^{\infty}\frac{(k+m)!}{ k!}B_k(t^{a}_1,\cdots,t^{a}_k)\frac{\partial}{\partial t_{k+m}}\bigg).
					\end{eqnarray*}
				\end{small} 
				Putting $\bar{m}=m+\gamma,\quad \bar{n}=n+\gamma,$ and by changing $n!\,{\partial\over \partial t_n}\longleftrightarrow x^n,$ we show directly that  the products $\mathcal{T}^{p^{a},q^{a}}_{m}\,.\mathcal{T}^{p^{b},q^{b}}_{n}$ and $\mathcal{T}^{p^{a},q^{a}}_{m}\,.\mathbb{T}^{p^{b},q^{b}}_{n}$ are respectively equivalent to (\ref{JSprod1}) and (\ref{JSprod2}).
					\subsection{{\bf Chakrabarti and Jagannathan} deformation \cite{CJ} }
					Setting $\mathcal{R}(x,y)={(1-xy)\over (p^{-1}-q)x},$ we deduce the $(p^{-1},q)-$ deformed super Virasoro $n-$ algebra and application.
		\subsection{{\bf Hounkonnou-Ngompe generalized $q-$ Quesne } deformation \cite{HN}}
				The results corresponding here are obtained by taking $\mathcal{R}(x,y)=\frac{(xy-1)}{ (q-p^{-1})y}.$
	\subsection{{\bf Biedenharn-Macfarlane } deformation \cite{B, M}}
										Putting ${\mathcal R}(x)={x-x^{-1} \over q-q^{-1}},$ we obtain the $q-$ deformed super Virasoro $n-$ algebra.
\section{Concluding and remarks}
We have constructed a super Witt $n$ and   Virasoro $2n-$ algebras from quantum algebras. Moreover, we have generalized this study to investigate  the super $\mathcal{R}(p,q)-$ deformed Witt $n-$ algebra, and  super $\mathcal{R}(p,q)-$ deformed Virasoro $n-$ algebra and discuss a toy model. Particular cases  have been investigated. For further, the super Virasoro algebra with a conformal dimenssion is in preparation for the futur work.
\section*{Acknowledgements}
This research was partly supported by the SNF Grant No. IZSEZ0\_206010.
%

%
%



\end{document}